\documentclass[11pt]{article}

\usepackage{bbm}
\usepackage{amsfonts}
\usepackage{amsmath}
\usepackage{amssymb}
\usepackage{amsthm}
\usepackage[mathscr]{eucal}
\usepackage{bbold}
\usepackage{braket}
\usepackage{color}
\usepackage{dsfont}
\usepackage{framed}
\usepackage{jheppub}
\usepackage{mathtools}
\usepackage{physics}
\usepackage[normalem]{ulem}
\usepackage{tensor}
\usepackage{thmtools}
\usepackage{thm-restate}
\usepackage{ulem}
\usepackage{tikz}
\usepackage{soul}

\usetikzlibrary{math} 

\definecolor{acolor}{RGB}{0,192,0}

\definecolor{gcolor}{RGB}{255,0,0}

\definecolor{rcolor}{RGB}{255,0,0}

\definecolor{bcolor}{RGB}{10,10,255}

\newcommand{\ignore}[1]{}

\newtheorem*{claim}{Claim}

\newtheorem*{defi}{Definition}

\newtheorem{nthm}{Theorem}[]
\newtheorem{nlemma}{Lemma}

\newtheorem{ndefi}{Definition}

\renewcommand{\vec}[1]{\boldsymbol{\mathbf{#1}}}
\renewcommand{\(}{\left(}
\renewcommand{\)}{\right)}
\renewcommand{\[}{\left[}
\renewcommand{\]}{\right]}

\makeatletter\def\@fpheader{~}\makeatother
\newcommand{\Eqref}[1]{Eq.~\eqref{#1}}
\newcommand{\secref}[1]{Sec.~\ref{#1}}
\newcommand{\figref}[1]{Fig.~\ref{#1}}
 
\usepackage[export]{adjustbox}
 \usepackage{simpler-wick}

\def\b{\beta}

\def\d{\delta}

\def\g{\gamma}

\def\r{\rho}

\def\y{\psi}

\def\er{\eqref}

\def\pa{\partial}

\def\({\left(}
\def\){\right)}
\def\[{\left[}
\def\]{\right]}
\def\<{\langle}
\def\>{\rangle}

\def\tr{\operatorname{tr}}

\renewcommand{\(}{\left(}
\renewcommand{\)}{\right)}
\renewcommand{\[}{\left[}
\renewcommand{\]}{\right]}

\renewcommand{\(}{\left(}
\renewcommand{\)}{\right)}
\renewcommand{\[}{\left[}
\renewcommand{\]}{\right]}
 \let\b=\beta \let\g=\gamma \let\d=\delta 
   
     \let\r=v
\let\sg=\sigma   

\let\y=\psi

\let\pa=\partial

\title{A Modified Cosmic Brane Proposal for Holographic Renyi Entropy}

\author[1]{Xi Dong,}
\emailAdd{xidong@ucsb.edu}

\author[2,3]{Jonah Kudler-Flam,}
\emailAdd{jkudlerflam@ias.edu}

\author[1,4]{and Pratik Rath}
\emailAdd{rath@ucsb.edu}

\affiliation[1]{Department of Physics, University of California, Santa Barbara, CA 93106, USA}
\affiliation[2]{School of Natural Sciences, Institute for Advanced Study, Princeton, NJ 08540 USA}
\affiliation[3]{Princeton Center for Theoretical Science, Princeton University, Princeton, NJ~08540, USA}
\affiliation[4]{Center for Theoretical Physics and Department of Physics,
University of California, Berkeley, CA 94720, USA}

\abstract{We propose a new formula for computing holographic Renyi entropies in the presence of multiple extremal surfaces. 
Our proposal is based on computing the wave function in the basis of fixed-area states and assuming a diagonal approximation for the Renyi entropy.
For Renyi index $n\geq1$, our proposal agrees with the existing cosmic brane proposal for holographic Renyi entropy.
For $n<1$, however, our modified cosmic brane proposal predicts a new phase with leading order (in Newton's constant $G$) corrections to the original cosmic brane proposal, even far from entanglement phase transitions and when bulk quantum corrections are unimportant.
Recast in terms of optimization over fixed-area states, the difference between the two proposals can be understood to come from the order of optimization: for $n<1$, the original cosmic brane proposal is a minimax prescription whereas our proposal is a maximin prescription.
We demonstrate the presence of such leading order corrections using illustrative examples.
In particular, our proposal reproduces existing results in the literature for the PSSY model and high-energy eigenstates, providing a universal explanation for previously found leading order corrections to the $n<1$ Renyi entropies. 
}

\begin{document}
 
 
\maketitle

\section{Introduction}

Entanglement plays a fundamental role in the emergence of spacetime in holographic theories of quantum gravity \cite{VanRaamsdonk:2010pw}.
The prime discovery establishing this connection is the Ryu-Takayanagi (RT) formula \cite{Ryu:2006bv,2006JHEP...08..045R,Hubeny:2007xt} in the AdS/CFT correspondence \cite{Maldacena:1997re}. It states that at leading order in Newton's constant $G$, we have
\begin{equation}
    S(R) = \frac{A\(\g_R\)}{4G},
\end{equation}
where $S(R)=-\tr(\rho_R \log \rho_R)$ is the entanglement entropy of the density matrix for a boundary subregion $R$, $A$~represents the area, and $\g_R$ is the RT surface, the bulk extremal surface anchored to $\pa R$ (and homologous to $R$) with minimal area. 
Quantum corrections to this formula are well understood \cite{Faulkner:2013ana,Engelhardt:2014gca} and play an important role in, e.g., black hole evaporation \cite{Penington:2019npb,Almheiri:2019psf,Penington:2019kki,Almheiri:2019qdq}.

The entanglement entropy belongs to a one-parameter family of Renyi entropies defined as
\begin{equation}
    S_n(R) = \frac{1}{1-n}\log \tr\(\rho_R^n\).
\end{equation}
Another related one-parameter family is that of the refined Renyi entropies \cite{Dong:2016fnf} 
 defined as
\begin{equation}\label{eq:refined}
    \tilde{S}_n(R) = n^2\pa_n\(\frac{n-1}{n}S_n(R)\).
\end{equation}
The entanglement entropy arises in the $n\to 1$ limit of either of these families.
Both $S_n(R)$ and $\tilde{S}_n(R)$ measure the bipartite entanglement between $R$ and its complementary subregion $\bar{R}$, and provide more detailed information than the entanglement entropy itself.
An arbitrary density matrix can be thought of as a thermal state for the modular Hamiltonian, $\rho = e^{ -H_{\textrm{mod.}}}$, in which case the Renyi and refined Renyi entropies essentially probe the system at different temperatures given by $\frac{1}{n}$.
Thus, it is of interest to understand the holographic dual of the Renyi entropy and the refined Renyi entropy to obtain a fine grained understanding of the entanglement spectrum. 

The Renyi entropies at integer $n$ can be computed using the replica trick in the boundary CFT \cite{Calabrese:2004eu}.
The insight of Ref.~\cite{Lewkowycz:2013nqa}, as we shall review in \secref{sub:cosmic}, was to propose that the corresponding dominant bulk gravitational saddle is replica symmetric.
Then, quotienting the bulk geometry by the replica symmetry, one obtains a geometry with an additional conical defect of opening angle $\frac{2\pi}{n}$ anchored to $\pa R$.
Such a conical defect could equivalently be interpreted as being induced by the insertion of a ``cosmic brane'' of appropriate tension \cite{Dong:2016fnf,Lewkowycz:2013nqa}.
This cosmic brane proposal provides a natural continuation away from integer $n$, and the RT formula follows from it in the limit $n\to1$.
Moreover, it was shown in Ref.~\cite{Dong:2016fnf} that the refined Renyi entropy is then computed by the area of the cosmic brane in this new spacetime, providing a generalization of the RT formula.
Henceforth, we will refer to either of the above proposals for Renyi entropy and refined Renyi entropy as the \textit{original cosmic brane proposal}.

In this paper, we will demonstrate that while the original cosmic brane proposal is correct in many situations, it can fail even at leading order (in $G$) in large regions of parameter space.
In particular, we will show that such corrections appear quite generically for Renyi index $n<1$ in the presence of multiple extremal surfaces.
Such corrections have previously been noticed by Refs.~\cite{Dong:2021oad,Akers:2022max} in the PSSY model of black hole evaporation \cite{Penington:2019kki}.
Moreover, the results of Ref.~\cite{Murthy:2019qvb} interpreted in a holographic context also imply such corrections for high-energy eigenstates.
We will present a \textit{modified cosmic brane proposal} that provides a universal explanation of such leading order corrections.

Our modified cosmic brane proposal is based on expanding the holographic state in a basis of fixed-area states where the areas of all extremal surfaces homologous to $R$ have small fluctuations \cite{Akers:2018fow,Dong:2018seb,Dong:2019piw}.
We review the decomposition of smooth holographic states into a fixed-area basis in \secref{sub:wave}, reformulating the original cosmic brane proposal in this language.
In the presence of two extremal surfaces $\gamma_1$ and $\gamma_2$, the original cosmic brane proposal (ignoring bulk quantum corrections) takes the form
\begin{equation}\label{eq:cosmic_A1}
    S_n^{C}(R) = \begin{cases}
        \displaystyle \frac{1}{1-n}\max_{i=1,2}\,\max_{A_1,A_2}\( n\log p\(A_1,A_2\) +\frac{(1-n)}{4G}A_i\)&n\geq 1,\\
        \displaystyle \frac{1}{1-n}\min_{i=1,2}\,\max_{A_1,A_2}\( n\log p\(A_1,A_2\) +\frac{(1-n)}{4G}A_i\)&n<1.
    \end{cases}
\end{equation}
where the function being optimized is the contribution to the $n$th Renyi entropy from the fixed-area saddle with areas $A_1,A_2$ where the replica gluing is performed around $\gamma_i$.
Here $p\(A_1,A_2\)$ is the probability distribution over areas of the extremal surfaces.

Fixed-area states provide a convenient basis for decomposing holographic entropy calculations.
Via their connection to random tensor networks \cite{Hayden:2016cfa}, they allow us to use other tools from quantum information theory to analyze holographic entanglement measures.
This connection has already been exploited to discover many new results for holography (see e.g.\ Refs.~\cite{Dong:2020iod,Marolf:2020vsi,Akers:2020pmf,2021PhRvL.126q1603K,Dong:2021oad,Akers:2021pvd,Akers:2022zxr,Akers:2023obn,akers2023reflected}).
Our findings in this paper provide another such example of taking inspiration from random tensor networks to learn about quantum gravity.

In \secref{sec:renyi}, we use the wave function in the fixed-area basis and assume a diagonal approximation to arrive at our modified cosmic brane proposal.
In summary, our modified cosmic brane proposal for the Renyi entropy (ignoring bulk quantum corrections) is
\begin{equation}\label{eq:resultshort}
    S_n^{MC}(R) = \displaystyle \frac{1}{1-n}\max_{A_1,A_2} \( n\log p\(A_1,A_2\) +\frac{(1-n)}{4G} \min\[A_1, A_2\]\).
\end{equation}
To compare with the original cosmic brane proposal \er{eq:cosmic_A1}, we rewrite \er{eq:resultshort} as
\begin{align}\label{eq:result}
    S_n^{MC}(R) = \begin{cases}
        \displaystyle \frac{1}{1-n}\max_{A_1,A_2} \,\max_{i=1,2}\( n\log p\(A_1,A_2\) +\frac{(1-n)}{4G}A_i\)&n\geq 1,\\
        \displaystyle \frac{1}{1-n}\max_{A_1,A_2} \,\min_{i=1,2}\( n\log p\(A_1,A_2\) +\frac{(1-n)}{4G}A_i\)&n<1.
    \end{cases}
\end{align}
Combining \Eqref{eq:result} with \Eqref{eq:refined} also leads to a modified cosmic brane proposal for the refined Renyi entropy,
\begin{equation}\label{eq:result2}
   \tilde{S}^{MC}_n\(R\) = \frac{\bar{A}_{\overline{i}}(n)}{4G}, 
\end{equation}
where $\{\bar{i},\bar{A_1}(n),\bar{A_2}(n)\}$ is the location of the optimum in \Eqref{eq:result} for a given value of $n$ and we have assumed a continuous probability distribution for \Eqref{eq:result2} to hold.
We expect our proposal \Eqref{eq:result}--\Eqref{eq:result2} to apply at leading order in $G$ for arbitrary $n$ when bulk quantum corrections can be ignored.
In some situations, it is also possible to understand bulk quantum corrections as we will discuss in \secref{sec:renyi}.

We can now compare our proposal to the original cosmic brane proposal.
The key difference between the two proposals \Eqref{eq:cosmic_A1} and \Eqref{eq:result} is the order of optimization.
For $n\geq1$, we have two maximizations whose order can be swapped and thus the two proposals always agree.
On the other hand, for $n<1$, the original cosmic brane proposal is a minimax prescription whereas the modified cosmic brane proposal is a maximin prescription.
Thus, in general, we can only conclude that $S_n^{MC}(R)\leq S_n^{C}(R)$, but they need not agree.

In \secref{sec:comp}, we provide necessary and sufficient conditions for an agreement between the two proposals.
The original cosmic brane proposal considered two candidate saddles, and each saddle is smooth everywhere except that it has a cosmic brane at $\gamma_i$ which sources a conical defect of opening angle $\frac{2\pi}{n}$.
We show that the two proposals agree whenever the cosmic brane sits at the minimal surface, i.e., when $A_i\leq A_{3-i}$ (evaluated in the corresponding saddle).
On the other hand, it is possible that neither of the saddles satisfies this constraint.
In such a case, the optimum in \Eqref{eq:result} may either be achieved by a subleading saddle in the original cosmic brane proposal or at the entanglement phase transition boundary $A_1=A_2$, either of which results in leading order corrections to the original cosmic brane proposal.
As we explain in more detail later, the configurations at $A_1=A_2$ involve distributing the cosmic brane tension over both candidate RT surfaces, and these new configurations were not considered in the original cosmic brane proposal.

To illustrate the presence of such corrections, we work out the example of $p\(A_1,A_2\)$ being a Gaussian distribution in \secref{sec:gaussian}.
In the simplest such setting, we demonstrate that the modified cosmic brane proposal agrees with the original cosmic brane proposal for $n\geq1$ as expected, whereas leading order corrections arise for $n<1$.
The detailed analysis of Renyi entropies in an arbitrary Gaussian distribution is provided in Appendix~\ref{app:gaussian}.

In \secref{sec:agree}, we provide evidence for our proposal by reproducing existing results in the literature for the holographic Renyi entropy.
We focus on results where leading order corrections (for $n<1$) have been previously found using methods different from ours.
We compute the Renyi entropies at arbitrary $n$ in the PSSY model in \secref{sub:PSSY} and in high-energy eigenstates in \secref{sub:eig}, precisely reproducing the known results in each of these cases with our modified cosmic brane proposal.

We discuss various future directions in \secref{sec:disc}.
A particular application of the modified cosmic brane proposal will be to compute the entanglement negativity in AdS/CFT, which we analyze in an accompanying paper \cite{neg}.
We also discuss the validity of our diagonal approximation and the possibility of replica symmetry breaking.

\section{Original Cosmic Brane Proposal}
\subsection{Review of the Proposal}
\label{sub:cosmic}

At integer $n$, the replica trick can be used to compute Renyi entropies for a subregion $R$ in the boundary CFT \cite{Calabrese:2004eu}.
The replica trick involves computing $\tr\(\rho_R^n\)$ in terms of a partition function $Z_n$ of the CFT involving $n$ copies of the original system glued cyclically about $\pa R$.
By the AdS/CFT dictionary \cite{Gubser:1998bc,Witten:1998qj}, in the saddle point approximation we have
\begin{equation}
    Z_n = \exp\(-I\[g_n\]\),
\end{equation}
where $I$ is the gravitational action and $g_n$ is a geometry that solves the equations of motion and satisfies the asymptotic boundary conditions defined by the replica trick path integral.
In this paper, we will restrict to discussing Einstein gravity, but we expect our results to generalize to higher-derivative theories using the ideas of Ref.~\cite{Dong:2013qoa,Dong:2019piw}.
At this point, we are ignoring bulk quantum corrections and they will be discussed briefly in \secref{sec:renyi}.

Ref.~\cite{Lewkowycz:2013nqa} proposed that the dominant saddle $g_n$ respects the $\mathbb{Z}_n$ replica symmetry of the boundary path integral that cyclically permutes the $n$ boundary copies. 
This insight allows one to quotient the bulk geometry by this $\mathbb{Z}_n$ symmetry and provides a natural continuation away from integer $n$, i.e., including the normalization factor we have
\begin{equation}\label{eq:LM}
    \tr\(\rho_R^n\) = \exp\left(-n (I\[\hat{g}_n\] - I\[{g}_1\]) \right),
\end{equation}
where $\hat{g}_n$ is a solution to the equations of motion with a conical defect of opening angle $\frac{2\pi}{n}$ anchored to $\pa R$ in addition to satisfying the asymptotic boundary conditions defining the original state.\footnote{The action $I\[\hat{g}_n\]$ excludes an explicit contribution from the conical defect \cite{Lewkowycz:2013nqa}.}
Equivalently, this conical defect can be interpreted as arising from the insertion of a cosmic brane of tension $T_n=\frac{n-1}{4nG}$ \cite{Lewkowycz:2013nqa}.
With this understanding, \Eqref{eq:LM} provides a natural continuation to non-integer values of $n$\footnote{Note that the continuation need not be analytic since there may generally be phase transitions in the $G\to0$ limit, in which case one cannot apply Carlson's theorem as one could for the boundary CFT at finite central charge.} and defines the original cosmic brane proposal for the holographic Renyi entropy for arbitrary $n$ at leading order in $G$, i.e.,
\begin{equation}\label{eq:LM2}
    S_n(R) = \frac{n}{n-1} \(I\[\hat{g}_n\]-I\[g_1\]\).
\end{equation}
In the limit $n\to1$, the cosmic brane becomes tensionless and the RT formula follows from \Eqref{eq:LM2}.

It was proposed in Ref.~\cite{Dong:2016fnf} that the refined Renyi entropy satisfies a natural generalization of the RT formula.
Using \Eqref{eq:LM2}, it was shown that 
\begin{equation}\label{eq:refined_holo}
    \tilde{S}_n\(R\) = \frac{A\(\g_{R,n}\)}{4G},
\end{equation}
where $\g_{R,n}$ is the location of the cosmic brane of tension $T_n$ in the geometry $\hat{g}_n$.

It is important to note that the naive version of the original cosmic brane proposal is already subtle for $n<1$ in the presence of multiple extremal surfaces that serve as candidate RT surfaces for subregion $R$.
When there are multiple extremal surfaces, at integer $n$ one naturally picks the solution with the least bulk action $I\[\hat{g}_n\]$ and it is natural to extend this rule to the continuation.
In the limit $n\to1^+$, this picks out the minimal area surface, resulting in the RT formula.

However, if we instead considered the limit $n\to1^-$, and if the original cosmic brane proposal still chooses the solution with the least action $I\[\hat{g}_n\]$, it would pick out the maximal area surface among the two candidates, leading to a physically unreasonable answer.
Thus, for $n<1$, it is natural to insist that the original cosmic brane proposal pick out the solution with the largest action instead.
With this rule, there is so far no obvious problem with the formula and we will take this to define the original cosmic brane proposal in the presence of multiple extremal surfaces.
Having said so, in this paper, we will demonstrate that even with this updated  rule, the original cosmic brane proposal can have corrections at leading order in $G$ in the presence of multiple extremal surfaces.

\subsection{Reformulation in Terms of Fixed-Area States}
\label{sub:wave}

Having discussed the original cosmic brane proposal in terms of the gravitational path integral, we will now reformulate it in terms of fixed-area states.
This will prove convenient later to compare with our modified proposal.

We start by considering a holographic state $\ket{\y}$ defined by a Euclidean path integral construction in the boundary CFT. 
For our purpose, we will be interested in considering a subregion $R$ of the boundary such that there are multiple candidate RT surfaces in the state $\ket{\y}$.
While we expect our formalism to go through in a straightforward manner for more than two surfaces, it suffices for illustrative purposes to restrict to having two extremal surfaces anchored to $\pa R$ (and homologous to $R$), labelled $\g_1$ and $\g_2$.

Following Refs.~\cite{Dong:2020iod,Marolf:2020vsi,Akers:2020pmf}, we will decompose the state $\ket{\y}$ into an orthonormal basis of fixed-area states $\ket{A_1,A_2}$ which, as we shall discuss, are states where the areas $A_{1,2}$ of surfaces $\g_{1,2}$ are sharply peaked \cite{Dong:2018seb,Akers:2018fow,Dong:2019piw}.
In this basis, we have
\begin{align}
\label{eq:decomposition}
    \ket{\y} = \sum_{A_1,A_2} \sqrt{p\(A_1,A_2\)} \,e^{i\varphi(A_1,A_2)}\ket{A_1,A_2},
\end{align}
where we have explicitly separated out the phases $\varphi(A_1,A_2)$ so that $p(A_1,A_2)$ is real and can be interpreted as the probability distribution of the areas of the two extremal surfaces.

Smooth holographic states (such as $\ket{\y}$) are defined by a path integral with asymptotic boundary conditions. 
Their corresponding fixed-area states are defined by an identical path integral with an additional boundary condition that fixes the area of the RT surface to be a specified value \cite{Akers:2018fow,Dong:2018seb,Dong:2019piw}. 
Doing so requires the opening angle to adjust in response and thus generically introduces a conical defect at the surface.

\begin{figure}[t]
\centering
\includegraphics[scale=0.7]{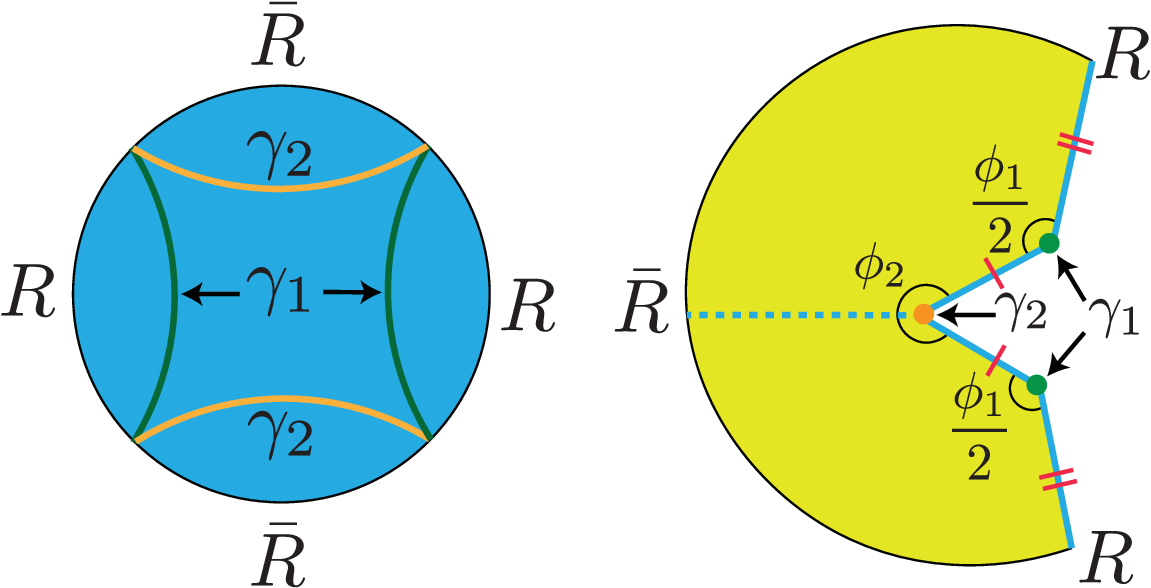}
\caption{(left): Cauchy slice of a geometry with subregion $R$ chosen to be two disjoint intervals. There are two candidate RT surfaces $\g_1$ (green) and $\g_2$ (orange). (right): The fixed-area Euclidean saddle has conical defects with opening angle $\phi_{1,2}$ at the extremal surfaces $\g_{1,2}$. The Cauchy slice is marked in blue, the bulk region to be cut open to obtain a single copy of $\rho_R\(A_1,A_2\)$ is solid, and the rest is dashed. \label{fig:fixed}}
\end{figure}

In our case of interest, there are two such candidate surfaces, $\g_1$ and $\g_2$. 
In general, $\g_1$ and $\g_2$ can be separately specified in a gauge invariant manner. 
For example, when $R$ consists of two intervals, we can specify them to be extremal surfaces of a given homotopy class, i.e., connected or disconnected with respect to $R$ (see \figref{fig:fixed}). 
In general, we will assume that $\g_1$ is the outermost surface, closest to $R$. 
Thus, our fixed-area saddles will satisfy all the asymptotic boundary conditions, satisfy the equations of motion everywhere away from $\g_1$ and $\g_2$, and have areas $A\(\g_1\)=A_1$ and $A\(\g_2\)=A_2$. 
As noted earlier, in general, this will introduce conical defects with opening angles $\phi_{1,2}$ at $\g_{1,2}$ as shown in \figref{fig:fixed}. 

The probability weights $p\(A_1,A_2\)$ can then be computed using the gravitational path integral. Namely, we have \cite{Marolf:2020vsi}
\begin{align}
\begin{aligned}
       p\(A_1,A_2\)&=\frac{\<\y|\Pi_{A_1,A_2}|\y\>}{\<\y|\y\>}\\
    &=\exp\(I\[g_{\y}\]-I\[g_{A_1,A_2}\]\),\label{eq:prob} 
\end{aligned}
\end{align}
where $\Pi_{A_1,A_2}$ is a projector onto areas $\(A_1,A_2\)$, $g_{A_1,A_2}$ is the corresponding fixed-area saddle and $g_{\y}$ is the original, non-fixed area saddle for $\ket{\psi}$.

For states prepared by simple path integrals, we expect $p\(A_1,A_2\)$ to be sufficiently smooth and peaked at a single value of $\(A_1,A_2\)$.
Such states are compressible in the sense defined in Ref.~\cite{Akers:2020pmf}.
Thus, although our discussion is amenable to the inclusion of incompressible states, the corrections to the Renyi entropy that we will find are independent from those discovered for the entanglement entropy of states with incompressible bulk matter.

Note that the above discussion can be extended to quantum extremal surfaces.
Fixed-area states can be defined in an analogous fashion for such surfaces.
For classical extremal surfaces, the extremality condition plays an important role in allowing the opening of a conical defect at the location of the surface as required to fix the area.
For quantum extremal surfaces, quantum extremality plays the same role allowing the equations of motion to be satisfied near the location of the defect due to a matter contribution that stabilizes the classical non-extremality \cite{Dong:2017xht}.

We are now ready to reformulate the original cosmic brane proposal in terms of fixed-area states.
\begin{ndefi}\label{CBP}
The original cosmic brane proposal for the Renyi entropy is given by
\begin{align}\label{eq:cosmic_A}
    S_n^{C}(R) = \begin{cases}
        \displaystyle \frac{1}{1-n} \max_{i=1,2}\,\max_{A_1,A_2}\( n\log p\(A_1,A_2\) +\frac{(1-n)}{4G}A_i\)&n\geq1,\\
        \displaystyle \frac{1}{1-n} \min_{i=1,2}\,\max_{A_1,A_2}\( n\log p\(A_1,A_2\) +\frac{(1-n)}{4G}A_i\)&n<1.
    \end{cases}
\end{align}
\end{ndefi}
To see that \Eqref{eq:cosmic_A} is equivalent to the formulation of the original cosmic brane proposal reviewed in \secref{sub:cosmic}, we can use \Eqref{eq:prob} to find the conditions for attaining the inner maximum in \Eqref{eq:cosmic_A}:
e.g. for $i=1$, we have
\begin{align}\label{eq:max1}
    \frac{\pa I\[g_{A_1,A_2}\]}{\pa A_1}&= \frac{1-n}{4 n G},\\
    \frac{\pa I\[g_{A_1,A_2}\]}{\pa A_2}&=0,\label{eq:max2}
\end{align}
and the roles of $A_1$ and $A_2$ are swapped for $i=2$.

In general, the action of the fixed-area saddle can be divided into two parts: a localized contribution from the surfaces $\g_{1,2}$ and the remaining contribution from the rest of the spacetime away from the surfaces \cite{Dong:2018seb}; i.e,
\begin{equation}
    I\[g_{A_1,A_2}\] = I_{away}\[g_{A_1,A_2}\]+\frac{\(\phi_1-2\pi\)A_1}{8\pi G}+\frac{\(\phi_2-2\pi\)A_2}{8\pi G},
\end{equation}
where we remind the reader that $\phi_i$ is the opening angle at surface $\g_i$ and thus, the localized contributions arise from Ricci scalar delta functions due to the presence of conical defects at $\g_{1,2}$.
To take the derivative with respect to $A_i$ for a smooth holographic state, we follow the argument of Ref.~\cite{Dong:2018seb} and note that differentiating $I_{away}\[g_{A_1,A_2}\]$ with respect to $\phi_i$ gives $-A_i/8\pi G$, and thus its Legendre transform $I\[g_{A_1,A_2}\]$ satisfies
\begin{equation}\label{eq:der}
    \frac{\pa I\[g_{A_1,A_2}\]}{\pa A_i} = \frac{\phi_i-2\pi}{8\pi G}.
\end{equation}

Using this, \Eqref{eq:max1} and \Eqref{eq:max2} imply that $\phi_1=\frac{2\pi}{n}$ and $\phi_2=2\pi$.
Again, by symmetry, we have $\phi_1=2\pi$ and $\phi_2=\frac{2\pi}{n}$ for $i=2$.
These are precisely the candidate saddles that one obtains from the original cosmic brane proposal, with \Eqref{eq:max1} indicating the presence of a cosmic brane of tension $T_n$ at $\g_i$, as discussed in \secref{sub:cosmic}.
Finally, as discussed in \secref{sub:cosmic}, the maximization (minimization) over $i$ for $n\geq1$ ($n<1$) results in the optimum configuration chosen in the original cosmic brane proposal.

While we showed \Eqref{eq:cosmic_A} to be equivalent to the original cosmic brane proposal for holographic states prepared by a smooth gravitational path integral, \Eqref{eq:cosmic_A} can be taken as the definition of the original cosmic brane proposal for more general probability distributions over fixed-area states.
In fact, this is natural since the effect of inserting a cosmic brane of tension $T_n$ is captured by \Eqref{eq:max1} and \Eqref{eq:max2} quite generally, even for non-smooth holographic states like fixed-area states.

\section{Modified Cosmic Brane Proposal}
\label{sec:renyi}

To obtain our modified cosmic brane proposal, we will directly use the wave function in the fixed-area basis to compute the Renyi entropy for subregion $R$. 
The density matrix $\rho_R$ is given by
\begin{align}
\begin{aligned}
\rho_R&=\sum_{A_1,A_2,A_1',A_2'}\sqrt{p\(A_1,A_2\) p\(A_1',A_2'\)}e^{i(\varphi(A_1,A_2)-\varphi(A_1',A_2'))}\tr_{\bar{R}}\(|A_1,A_2\>\<A_1',A_2'|\)\\
&= \sum_{A_1,A_2} p\(A_1,A_2\) \rho_R\(A_1,A_2\)+OD_R,\label{eq:rho_diag}    
\end{aligned}
\end{align}  
where we have separated the diagonal part of \Eqref{eq:rho_diag} (i.e., terms with $A_1=A_1'$ and $A_2=A_2'$) from the off-diagonal part, represented by $OD_R$.
The diagonal pieces $\rho_R\(A_1,A_2\)$ are just the density matrices for fixed-area states, which are given by cutting open the Euclidean fixed-area saddle depicted in \figref{fig:fixed}.

In AdS/CFT, it is well understood that any operator in the entanglement wedge, the region between $R$ and the RT surface $\g_R$, can be measured by subregion $R$ \cite{Almheiri:2014lwa,Dong:2016eik}.
Similarly, operators in the complementary entanglement wedge can be measured by subregion $\bar{R}$.
In this setup, the RT surface is either $\g_1$ or $\g_2$ depending on which one has minimal area.
Irrespective of which surface is the true RT surface, $A_2$ is always measurable by $\bar{R}$.
Thus, we can ignore any off-diagonal terms in \Eqref{eq:rho_diag} where $A_2\neq A_2'$ since they vanish when performing the trace over $\bar{R}$ \cite{Marolf:2020vsi,Akers:2020pmf}. 
We can also ignore terms where $A_1\neq A_1'$ as long as $A_1<A_2$ for the same reason. 
However, we cannot in general ignore off-diagonal terms where $A_1\neq A_1'$ and $A_1>A_2$.

Despite this, Ref.~\cite{Marolf:2020vsi} found reasonable results by assuming a diagonal approximation for the moments of $\rho_R$, i.e.,
\begin{equation}\label{eq:diag1}
    \tr\(\rho_R^n\)\approx \sum_{A_1,A_2} p\(A_1,A_2\)^n \tr\(\rho_R\(A_1,A_2\)^n\).
\end{equation}
While this approximation can be justified for the entanglement entropy \cite{Akers:2020pmf}, there is no similar justification that we can provide for the Renyi entropies in general.
Nonetheless, we will proceed by assuming this diagonal approximation.
At leading order in $G$, we will further assume that $p(A_1,A_2)$ is such that we can apply \Eqref{eq:diag1} in the saddle point approximation, i.e.,
\begin{equation}\label{eq:diag2}
    \tr\(\rho_R^n\)\approx \max_{A_1,A_2} \[ p\(A_1,A_2\)^n \tr\(\rho_R\(A_1,A_2\)^n\)\],
\end{equation}
which generally only leads to $\log G$ corrections to the entropy compared to \Eqref{eq:diag1}.\footnote{Near entanglement phase transitions, it can lead to $O\(\frac{1}{\sqrt{G}}\)$ corrections to the entanglement entropy.} 
Our philosophy will be to derive our modified proposal under this assumption and provide evidence for the assumption in \secref{sec:agree} by demonstrating agreement with known results in AdS/CFT, derived using independent methods.

Note also that the diagonal terms considered in \Eqref{eq:diag1} are precisely the ones with replica-symmetric boundary conditions (viewing the area fixing constraint as a boundary condition).
As we note below, replica-symmetry breaking saddles can still contribute to the diagonal terms and we shall include them in our analysis.
The off-diagonal terms will generically involve only replica-symmetry breaking saddles.
In this sense, our assumption of ignoring off-diagonal terms is weaker than the assumption of Ref.~\cite{Lewkowycz:2013nqa}.
In fact, at leading order in $G$, one can instead directly assume \Eqref{eq:diag2}.
Picking the dominant diagonal term is similar to the assumption of Ref.~\cite{Lewkowycz:2013nqa} which assumed that at integer $n$, the dominant saddle in the Renyi entropy computation is replica-symmetric.
The advantage of directly assuming \Eqref{eq:diag2} would be that perhaps one can prove it in other ways, e.g., by using properties of the gravitational path integral. 
We will comment more on this in \secref{sec:disc}.
For now, we note that the diagonal approximation, as well as the saddle point approximation to it, satisfies the basic consistency check of being symmetric between the two systems, i.e, $S_n\(R\)=S_n\(\bar{R}\)$.

\begin{figure}[t]
\centering
\includegraphics[scale=0.7]{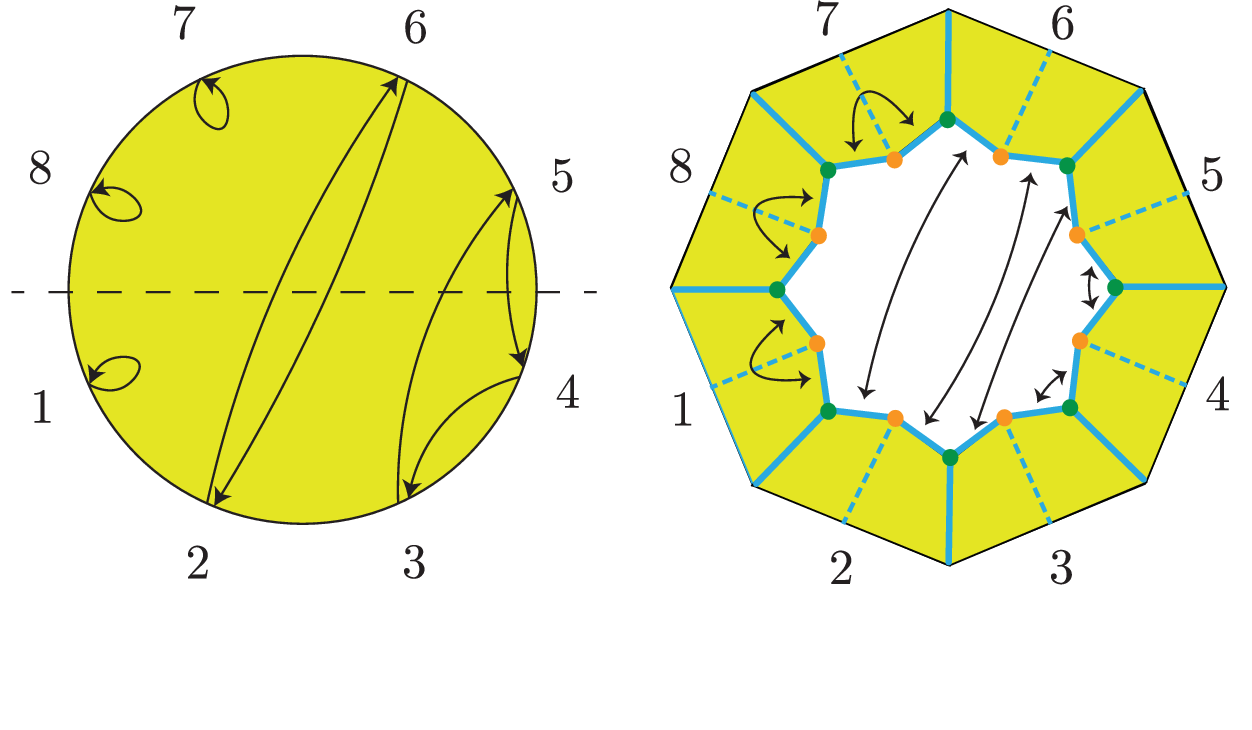}
\caption{(left): A permutation on $n$ elements. (right): The corresponding fixed-area Euclidean saddle that computes $\tr\(\rho_R\(A_1,A_2\)^n\)$ involves $n$ copies of the original saddle (see \figref{fig:fixed}) glued together in a manner dictated by the permutation.\label{fig:gluing}}
\end{figure}

In order to compute \Eqref{eq:diag2}, we need to understand how to compute $\tr\(\rho_R^n\)$ for a fixed-area state.
This was understood in detail in Refs.~\cite{Dong:2018seb,Akers:2018fow,Dong:2019piw} and we briefly review it here.
For integer $n$, $\tr\(\rho_R\(A_1,A_2\)^n\)$ is a partition function with $n$ copies of the original system that are required to satisfy a fixed-area boundary condition at the extremal surfaces.
The saddles that compute $\tr\(\rho_R\(A_1,A_2\)^n\)$ are given by cutting open the original fixed-area geometry, taking $n$ copies of it, and gluing them together in a manner corresponding to an arbitrary permutation of $n$ objects (see \figref{fig:gluing}).
For $A_1<A_2$, the dominant such saddle $g_n$ corresponds to the identity permutation, while for $A_1>A_2$, $g_n$ corresponds to the cyclic permutation.
Momentarily ignoring bulk quantum corrections, the classical action for the dominant saddle is given by \cite{Dong:2018seb}
\begin{equation}\label{eq:act}
    I\[g_n\]=n I\[g_1\]+\frac{n-1}{4G}\min\[A_1,A_2\].
\end{equation}
At $A_1=A_2$, a class of replica symmetry breaking saddles becomes degenerate, but since the degeneracy only contributes at subleading order, for our purpose we can continue using \Eqref{eq:act} even at $A_1=A_2$.
Accounting for normalization, at leading order, one obtains
\begin{align}
\begin{aligned}
    \tr\(\rho_R\(A_1,A_2\)^n\)&= \exp\(-I\[g_n\]+nI\[g_1\]\)\\
    &= \exp\(-\frac{n-1}{4G}\min\[A_1,A_2\]\)\label{eq:rhoRfixed}.
\end{aligned}
\end{align}

While we derived \Eqref{eq:rhoRfixed} for integer $n$, the result can be continued to arbitrary $n$ in the obvious way.
At (positive) integer $n$ (as well as real $n > 1$), the minimization over areas arises automatically from minimizing the action, although this action minimization naively appears to turn into an area maximization for $n<1$ as discussed in \secref{sub:cosmic}.
Nevertheless, fixed-area states are good analogues of states in random tensor networks, and in that context, \Eqref{eq:rhoRfixed} is known to be correct even for $n<1$.\footnote{An easy way to see this is to notice that the minimal cut in a tensor network puts a bound on the rank of the density matrix, which is the exponential of the $n=0$ Renyi entropy. Using monotonicity of Renyi entropy as a function of $n$ then constrains the Renyi entropies for $n<1$ to be given by the minimal cut. More generally, the precise spectrum can be derived even in the large bond dimension limit using the method of moments.}
This explicit minimization over areas is crucial for our proposal.

We can now combine \Eqref{eq:rhoRfixed} with \Eqref{eq:diag2} to compute the Renyi entropy at leading order, leading to our main result.
\begin{defi}
The modified cosmic brane proposal for Renyi entropy is given by
\begin{align}\label{eq:renyi}
    S_n^{MC}(R) = \begin{cases}
        \displaystyle \frac{1}{1-n}\max_{A_1,A_2}\, \max_{i=1,2}\( n\log p\(A_1,A_2\) +\frac{(1-n)}{4G}A_i\)&n\geq1,\\
        \displaystyle \frac{1}{1-n}\max_{A_1,A_2}\, \min_{i=1,2}\( n\log p\(A_1,A_2\) +\frac{(1-n)}{4G}A_i\)&n<1.
    \end{cases}
\end{align}
\end{defi}
Using \Eqref{eq:renyi}, one can also obtain the refined Renyi entropy \Eqref{eq:refined}.
Denoting $\{\overline{i},\bar{A_1}(n),\bar{A_2}(n)\}$ as the location where the optimum in \Eqref{eq:renyi} is achieved at a given $n$, we then have
\begin{equation}\label{eq:refrenyi}
    \tilde{S}_n^{MC}(R) = \frac{\bar{A_{\overline{i}}}(n)}{4G}.
\end{equation}
The simplest way of deriving \Eqref{eq:refrenyi} is to assume $p\(A_1,A_2\)$ to be differentiable so that the terms arising from $n$-derivatives acting on $\bar{A_i}(n)$ cancel out due to stationarity.
More generally, \Eqref{eq:refrenyi} can also be proved for continuous $p\(A_1,A_2\)$ using Danskin's theorem (see e.g. Ref.~\cite{bertsekas1997nonlinear}).

Note that \Eqref{eq:refrenyi} is applicable within a given phase where $\bar{A_{\overline{i}}}(n)$ changes continuously, although it can jump discontinuously at phase boundaries.
The Renyi entropy, on the other hand, is typically continuous but not analytic at such phase boundaries.

We can now discuss the effect of bulk quantum corrections to our modified cosmic brane proposal.
In general, accounting for bulk quantum corrections is complicated since the replica-symmetry breaking contributions need to be carefully summed over in order to continue to non-integer $n$.
We will discuss this further in \secref{sec:disc}.

\begin{figure}[t]
\centering
\includegraphics[scale=0.7]{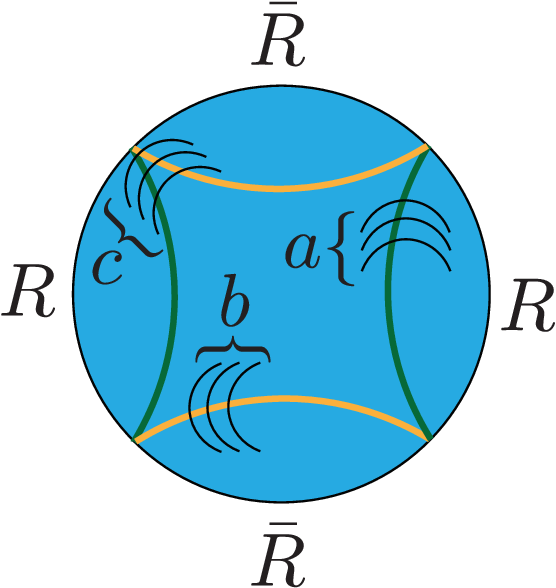}
\caption{A bulk state with only bipartite entanglement between the relevant bulk regions. The black lines connecting various regions represent Bell pairs of three different types labelled $a$, $b$ and $c$. \label{fig:bulk}}
\end{figure}

A simple update to \Eqref{eq:renyi} to include some bulk quantum corrections arises if we can find a code subspace such that the bulk entropy contributions from the entanglement wedges corresponding to both $\gamma_1$ and $\gamma_2$ can be rewritten (at least approximately) as the expectation values of commuting linear operators.
If so, the bulk entropy terms can be simply absorbed into a redefinition of the area operators, and we can derive \Eqref{eq:renyi} using the new area operators (assuming that replica-symmetry breaking contributions continue to be subleading).

A simple case where this can be done is when the bulk state is bipartitely entangled at leading order in $G$ as shown in \figref{fig:bulk}.
In this case, the leading order (in $G$) bulk entropies can be absorbed into a redefinition of the area operators. The new area operators commute and can be used to define a natural generalization of fixed-area states. Each such fixed-area state contains three kinds of Bell pairs: $a$, $b$ and $c$ as shown in \figref{fig:bulk}.
In such a state, the original areas and the redefined areas are related by $\frac{A_1}{4G}\to \frac{A_1}{4G} + S_a + S_c$ and $\frac{A_2}{4G}\to \frac{A_2}{4G} + S_b + S_c$, where $S_a$, $S_b$, $S_c$ are the entropies of one half of the $a$, $b$, $c$ Bell pairs.
Using the new area operators, we then find our modified cosmic brane proposal \Eqref{eq:renyi}.

\section{Comparing the Two Proposals}
\label{sec:comp}

Having formulated both proposals in terms of an optimization over fixed-area geometries, we can now compare them and see when they agree.
For starters, the difference between the proposals \Eqref{eq:cosmic_A} and \Eqref{eq:renyi} is the order of optimizations.
This leads us to the following general comparison:
\begin{nthm}
    For $n\geq1$, $S_n^{MC}(R)=S_n^{C}(R)$. For $n<1$, $S_n^{MC}(R)\leq S_n^{C}(R)$.
\end{nthm}
\begin{proof}
For $n\geq1$, we have two maximizations whose order can always be swapped.
Thus, the two proposals agree in this case.
For $n<1$, the original cosmic brane proposal is a minimax prescription whereas the modified cosmic brane proposal is a maximin prescription.
From the well-known max-min inequality (see for instance Ref.~\cite{boyd2004convex}), we obtain the required inequality.
\end{proof}

We will now find necessary and sufficient conditions for the two proposals to agree for $n<1$.
The original cosmic brane proposal has two candidate saddles, each with a cosmic brane sitting at $\gamma_i$.
We will find that the two proposals agree if and only if at least one of these saddles satisfies the constraint that $\gamma_i$ is the minimal surface among $\gamma_{1,2}$ (which we will refer to as the minimality constraint).

To do so, it is useful to establish some notation.
Let $f_i\equiv n \log p\(A_1,A_2\)+(1-n)\frac{A_i}{4G}$ for $i=1,2$.
Let $\vec{A^{(i)}}=\(A_1^{(i)},A_2^{(i)}\)$ be the location\footnote{Location refers to a point in the $\(A_1,A_2\)$ parameter space.} that maximize $f_i$ subject to the minimality constraint $A_i\leq A_{3-i}$.
Further define $\vec{\tilde{A}^{(i)}}=\(\tilde{A}_1^{(i)},\tilde{A}_2^{(i)}\)$ to be the location that maximizes $f_i$ without any constraint.
In cases with multiple maxima, we simply choose $\vec{A^{(i)}}$ to be any of the maximal locations, and choose $\vec{\tilde{A}^{(i)}}$ to be $\vec{A^{(i)}}$ if it satisfies the minimality constraint, and if otherwise, choose it to be any of the allowed maximal locations.
In the above notation, all the locations depend on $n$, which we leave implicit in this section.

With this notation, we can rewrite the two proposals in the $n<1$ case as 
\begin{claim}
The original cosmic brane proposal for $n<1$ is given by
\begin{align}\label{eq:CS}
    S_n^C(R) = \frac{1}{1-n}\min_{i=1,2}f_i\(\vec{\tilde{A}^{(i)}}\).
\end{align}
\end{claim}
\begin{claim}
The modified cosmic brane proposal for $n<1$ is given by
\begin{align}\label{eq:MCS}
    S_n^{MC}(R) = \frac{1}{1-n} \max_{i=1,2} f_i\(\vec{A^{(i)}}\).
\end{align}
\end{claim}
We will now determine necessary and sufficient conditions under which these proposals agree.

\begin{nlemma}\label{lem4}
For $n<1$, if $\vec{A^{(1)}}=\vec{\tilde{A}^{(1)}}$, then $S_n^C(R)=\frac{1}{1-n}f_1\(\vec{\tilde{A}^{(1)}}\)$.
An analogous statement holds with the roles of $\vec{A^{(1)}}$ and $\vec{A^{(2)}}$ reversed.
\end{nlemma}
\begin{proof}
    \begin{align}\label{eq:proof1}
        f_1\(\vec{\tilde{A}^{(1)}}\) = f_1\(\vec{A^{(1)}}\) &= n \log p\(A_1^{(1)},A_2^{(1)}\)+(1-n)\frac{A_1^{(1)}}{4G}\\\label{eq:proof12}
        &\leq n \log p\(A_1^{(1)},A_2^{(1)}\)+(1-n)\frac{A_2^{(1)}}{4G}=f_2\(\vec{A^{(1)}}\)\\
        &\leq f_2\(\vec{\tilde{A}^{(2)}}\),
    \end{align}
where the second line uses the fact that $\vec{A^{(1)}}$ by definition lies within the constrained domain $A_1\leq A_2$, together with the fact that $n<1$.
The third line uses the fact that $\vec{\tilde{A}^{(2)}}$ is the unconstrained maximum of $f_2$.
Therefore, $S_n^C\(R\) = \frac{1}{1-n}f_1\(\vec{\tilde{A}^{(1)}}\)$ due to the minimization in \Eqref{eq:CS}.
\end{proof}

\begin{nlemma}\label{lem5}
For $n<1$, if $\vec{A^{(1)}}=\vec{\tilde{A}^{(1)}}$, then $S_n^{MC}(R)=\frac{1}{1-n}f_1\(\vec{A^{(1)}}\)$.
An analogous statement holds with the roles of $\vec{A^{(1)}}$ and $\vec{A^{(2)}}$ reversed.
\end{nlemma}
\begin{proof}
\begin{align}
        f_1\(\vec{A^{(1)}}\) &= f_1\(\vec{\tilde{A}^{(1)}}\)\\
        &\geq f_1\(\vec{A^{(2)}}\)
        = n \log p\(A_1^{(2)},A_2^{(2)}\)+(1-n)\frac{A_1^{(2)}}{4G}\\
        &\geq n \log p\(A_1^{(2)},A_2^{(2)}\)+(1-n)\frac{A_2^{(2)}}{4G}=f_2\(\vec{A^{(2)}}\),
    \end{align}
where the second line uses the fact that $\vec{\tilde{A}^{(1)}}$ is the unconstrained maximum of $f_1$.
The third line uses the fact that $\vec{A^{(2)}}$ is constrained to lie within the domain $A_2\leq A_1$, in addition to $n<1$.
Thus, we see that $S_n^{MC}(R)=\frac{1}{1-n}f_1\(\vec{A^{(1)}}\)$ since the maximization in the modified cosmic brane proposal \Eqref{eq:MCS} is achieved at $\vec{A^{(1)}}$.
\end{proof}

\begin{nlemma}\label{lem6}
If $S_n^{MC}(R)=S_n^{C}(R)$, then $\vec{A^{(1)}}=\vec{\tilde{A}^{(1)}}$ or $\vec{A^{(2)}}=\vec{\tilde{A}^{(2)}}$.
\end{nlemma}
\begin{proof}
By definition, $f_i\(\vec{A^{(i)}}\)=f_i\(\vec{\tilde{A}^{(i)}}\)$ implies that $\vec{A^{(i)}}=\vec{\tilde{A}^{(i)}}$. Thus, we can instead prove that either $f_1\(\vec{A^{(1)}}\)=f_1\(\vec{\tilde{A}^{(1)}}\)$ or $f_2\(\vec{A^{(2)}}\)=f_2\(\vec{\tilde{A}^{(2)}}\)$. $S_n^{MC}(R)=S_n^{C}(R)$ implies that $f_{i_C}\(\vec{\tilde{A}^{(i_C)}}\)=f_{i_{MC}}\(\vec{A^{(i_{MC})}}\)$, where $i_C$ ($i_{MC}$) is the index that achieves the minimum (maximum) in \Eqref{eq:CS} (\Eqref{eq:MCS}). 

If $i_C=i_{MC}$, then we are done. If not, consider $i_C=1$ and $i_{MC}=2$ without loss of generality. Then we have $f_1\(\vec{\tilde{A}^{(1)}}\)=f_2\(\vec{A^{(2)}}\)$. However, we also have
\begin{equation}\label{eq:pf2}
    f_2\(\vec{A^{(2)}}\)\leq f_1\(\vec{A^{(2)}}\)\leq f_1\(\vec{\tilde{A}^{(1)}}\),
\end{equation}
where the first inequality is analogous to the one shown in \eqref{eq:proof12} and the second inequality follows from the definition of $\vec{\tilde{A}^{(1)}}$. Thus, equality must hold throughout and in particular, $f_2\(\vec{A^{(2)}}\)= f_1\(\vec{A^{(2)}}\)$, which in turn implies that $A_1^{(2)}=A_2^{(2)}$. Thus, $\vec{A^{(2)}}$ lies within the constrained region defining $\vec{A^{(1)}}$. Hence, we have $f_1\(\vec{A^{(2)}}\)\leq f_1\(\vec{A^{(1)}}\)$. However, since $i_{MC}=2$, we also have $f_2\(\vec{A^{(2)}}\)\geq f_1\(\vec{A^{(1)}}\)$. Combining this with the equality $f_2\(\vec{A^{(2)}}\)= f_1\(\vec{A^{(2)}}\)$ obtained above, we have $f_1\(\vec{A^{(1)}}\)=f_1\(\vec{A^{(2)}}\)$. This implies that $f_1\(\vec{A^{(1)}}\)=f_1\(\vec{\tilde{A}^{(1)}}\)$ as needed. 
\end{proof}

Combining Lemma~\ref{lem4}, Lemma~\ref{lem5} and Lemma~\ref{lem6}, we immediately find the following necessary and sufficient condition for the two proposals to agree:
\begin{nthm}
    $S_n^{MC}(R)=S_n^{C}(R)$ iff $\vec{A^{(1)}}=\vec{\tilde{A}^{(1)}}$ or $\vec{A^{(2)}}=\vec{\tilde{A}^{(2)}}$, i.e., if either of the saddles $\vec{\tilde{A}^{(i)}}$ satisfies the minimality constraint $A_i\leq A_{3-i}$.
\end{nthm}

Thus, we see that the two proposals agree as long as at least one of the candidate saddles in the original cosmic brane proposal satisfies the minimality constraint.
This implies that the two proposals can only disagree when neither of these saddles satisfies the minimality constraint.
In such a situation, the modified cosmic brane proposal will pick either a subleading saddle (for the original cosmic brane proposal) that satisfies the minimality constraint, or a diagonal saddle $\vec{A^{(D)}}$ that satisfies $A_1=A_2$.\footnote{
Related boundary value dominance has shown up in similar analyses of logarithmic negativity \cite{2022PhRvL.129f1602V, 2021arXiv211200020V}.}
We will explore various controlled examples in \secref{sec:gaussian} and \secref{sec:agree} where $\vec{A^{(D)}}$ does indeed dominate for $n<1$, resulting in leading order corrections to the original cosmic brane proposal.
We expect this to be a generic feature for $n<1$.

Before moving on, it is illuminating to discuss the geometry of the saddle $\vec{A^{(D)}}$ to contrast with the original cosmic brane proposal in situations where the state is defined by a smooth gravitational path integral.
On the diagonal, we can write $\min\[A_1,A_2\] = x A_1 + (1-x) A_2$ for any $x$.
Then, the maximizing conditions for the modified cosmic brane proposal are
\begin{align}\label{eq:max3}
     \frac{\pa I\[g_{A_1,A_2}\]}{\pa A_1}&= \frac{x(1-n)}{4 n G},\\
    \frac{\pa I\[g_{A_1,A_2}\]}{\pa A_2}&= \frac{(1-x)(1-n)}{4 n G}\label{eq:max4},
\end{align}
where we look for solutions that satisfy the constraint $A_1=A_2$, thus providing enough conditions to solve for $x$.
The conditions \Eqref{eq:max3} and \Eqref{eq:max4} can be interpreted as introducing cosmic branes with distributed tensions $T_{n,1}=x T_n$ and $T_{n,2}=(1-x)T_n$ at the surfaces $\g_1$ and $\g_2$ respectively.
This is in contrast with the original cosmic brane proposal where the cosmic brane of tension $T_n$ lies at a single surface.

\section{An Illustrative Example: Gaussian Distribution}
\label{sec:gaussian}
To illustrate how our modified cosmic brane proposal works, we consider the example of a Gaussian probability distribution\footnote{While it is physically reasonable to restrict the area distribution to be supported only on the domain $A_1, A_2 \geq0$, we will work in regimes far from these boundaries.
For instance, we will see this to be the case as long as $A_{1,0},A_{2,0}\gg \frac{\sg^2}{G}$.
Thus, the extension of the probability distribution to negative areas will not affect our discussion.
For concreteness, one could also truncate the distribution to the positive area domain, which does not affect any of our calculations at leading order. 
} over the areas such that
\begin{equation}\label{eq:gengauss}
    p\(A_1,A_2\) = \exp\[-\frac{1}{2}\(\vec{A}-\vec{A}_0\)\cdot C^{-1}\cdot \(\vec{A}-\vec{A}_0\)\],
\end{equation}
where $\vec{A}=\(A_1,A_2\)$ as before, $\vec{A}_0=\(A_{1,0},A_{2,0}\)$ represents the most likely area vector, and $C$ is the covariance matrix given by
\begin{equation}\label{eq:Cgen}
    C=\begin{pmatrix} 
  \sg_1^2     & r \sg_1 \sg_2\\ 
  r \sg_1 \sg_2 & \sg_2^2
\end{pmatrix},
\end{equation}
where $r\in\[-1,1\]$.
In this section, we present the simplest case that illustrates our point: $r=0$ and $\sg_1=\sg_2=\sg$.
A complete analysis of the general case \Eqref{eq:Cgen} is relegated to Appendix~\ref{app:gaussian}.

\begin{figure}[t]
\centering
\includegraphics[scale=0.6]{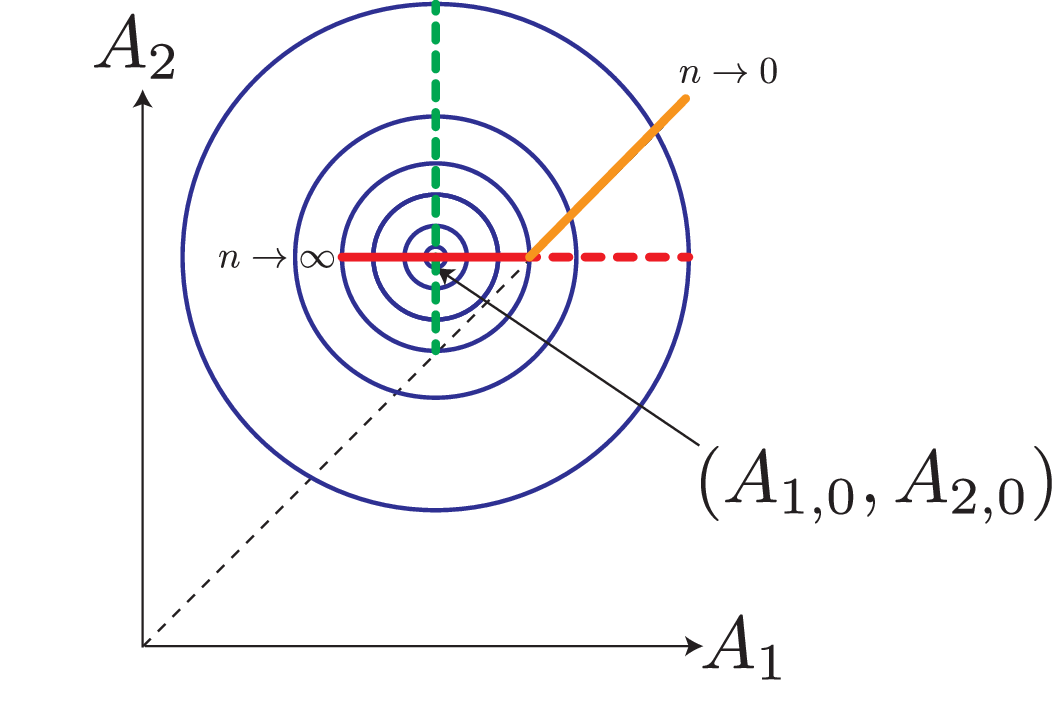}
\caption{The locus of maxima computing the Renyi entropy using the modified cosmic brane proposal is shown as a function of $n$ (solid red and orange lines), with the extremes $n\to \infty$ and $n\to0$ labelled. Level sets of constant probability are circles (blue) with center at $\(A_{1,0},A_{2,0}\)$. The cosmic brane saddle for $i=1$ (red) exits the allowed domain for $n<n_*$ as depicted by the dashed red line. The cosmic brane saddle for $i=2$ (dashed green) is always in the disallowed region in the above setup.}
\label{fig:gaussian1}
\end{figure}

For this setup, we can analytically compute the Renyi entropy using the modified cosmic brane proposal \Eqref{eq:renyi}.
There are three potential maxima that are relevant in the computation: $\vec{\tilde{A}^{(1)}}, \vec{\tilde{A}^{(2)}}$ which represent the unconstrained maxima of $f_1, f_2$ as before, and $\vec{A^{(D)}}$ which represents the maxima on the diagonal $A_1=A_2$.
Hereafter, we make the $n$ dependence in $\vec{\tilde{A}^{(i)}}$ and $\vec{A^{(D)}}$ explicit.

The maximizing conditions for the modified cosmic brane proposal then become a condition on the gradient \begin{equation}\label{eq:max5}
    \nabla \log p = \begin{pmatrix} 
  \frac{n-1}{4 n G} \\ 
  0
\end{pmatrix},
\end{equation}
for $\vec{\tilde{A}^{(1)}}$ and
\begin{equation}\label{eq:max6}
    \nabla \log p = \begin{pmatrix} 
  0 \\ 
  \frac{n-1}{4 n G}
\end{pmatrix},
\end{equation}
for $\vec{\tilde{A}^{(2)}}$.
Solving \Eqref{eq:max5} and \Eqref{eq:max6} we obtain 
\begin{align}\label{eq:A1max}
    \vec{\tilde{A}^{(1)}}(n) &= \vec{A}_0 - \frac{n-1}{4 G n}\begin{pmatrix} 
  \sg^2 \\ 
  0
\end{pmatrix},\\
\vec{\tilde{A}^{(2)}}(n) &= \vec{A}_0 - \frac{n-1}{4 G n}\begin{pmatrix} 
  0 \\ 
  \sg^2 
\end{pmatrix}.\label{eq:A2max}
\end{align}
We will refer to phases where the above peaks dominate as Phase 1 and Phase 2 respectively.
The Renyi entropy and refined Renyi entropy in these phases are given by
\begin{align}
    &S_n^{(1)}(R) = \frac{A_{1,0}}{4G} + \frac{(1-n)\sg^2}{32 n G^2}, &\tilde{S}_n^{(1)}(R) = \frac{A_{1,0}}{4G} + \frac{(1-n)\sg^2}{16 n G^2},\\
    &S_n^{(2)}(R) = \frac{A_{2,0}}{4G} + \frac{(1-n)\sg^2}{32 n G^2}, &\tilde{S}_n^{(2)}(R) = \frac{A_{2,0}}{4G} + \frac{(1-n)\sg^2}{16 n G^2}.
\end{align}

Similarly, solving for the location of the diagonal peak $\vec{A^{(D)}}(n)$, we obtain
\begin{equation}\label{eq:Admax}
    A_1^{(D)}(n)=A_2^{(D)}(n)=\frac{A_{1,0}+A_{2,0}}{2}-\frac{(n-1)\sg^2}{8nG}.
\end{equation}
The phase where the diagonal peak dominates will henceforth be called Phase D.
In Phase D, the Renyi entropy and refined Renyi entropy are given by
\begin{align}
    &S_n^{(D)}(R) = \frac{A_{1,0}+A_{2,0}}{8G} + \frac{n \(\Delta A_0\)^2}{4(n-1)\sg^2}-\frac{(n-1)\sg^2}{64 n G^2},&\quad\tilde{S}_n^{(D)}(R) = \frac{A_{1,0}+A_{2,0}}{8G}-\frac{(n-1)\sg^2}{32 n G^2},
\end{align}
where $\Delta A_0 := A_{2,0}-A_{1,0}$.

Having computed the Renyi entropy at each of the candidate peaks, we can now discuss which phases dominate.
Without loss of generality, consider the case where $\vec{A}_0$ lies in the domain $A_{1,0}<A_{2,0}$.
Then, it is easy to see that for $n<1$, the saddle at $\vec{\tilde{A}^{(1)}}(n)$ exits the allowed domain at a critical value
\begin{equation}\label{eq:n*}
    n_* = \frac{1}{1+\frac{4 G \Delta A_0}{\sg^2}},
\end{equation}
where $\Delta A_0 = A_{2,0}-A_{1,0}>0$ and thus, $n_*\in\[0,1\]$.
It is also easy to see from \Eqref{eq:A2max} that the saddle at $\vec{\tilde{A}^{(2)}}(n)$ is not allowed for any $n<1$.
Thus, the maximum is achieved at the diagonal peak $\vec{A^{(D)}}(n)$ for $n<n_*$.
The locus of maxima computing the Renyi entropy at different values of $n$ are shown in \figref{fig:gaussian1}.

\begin{figure}[t]
\centering
\includegraphics[scale=0.4]{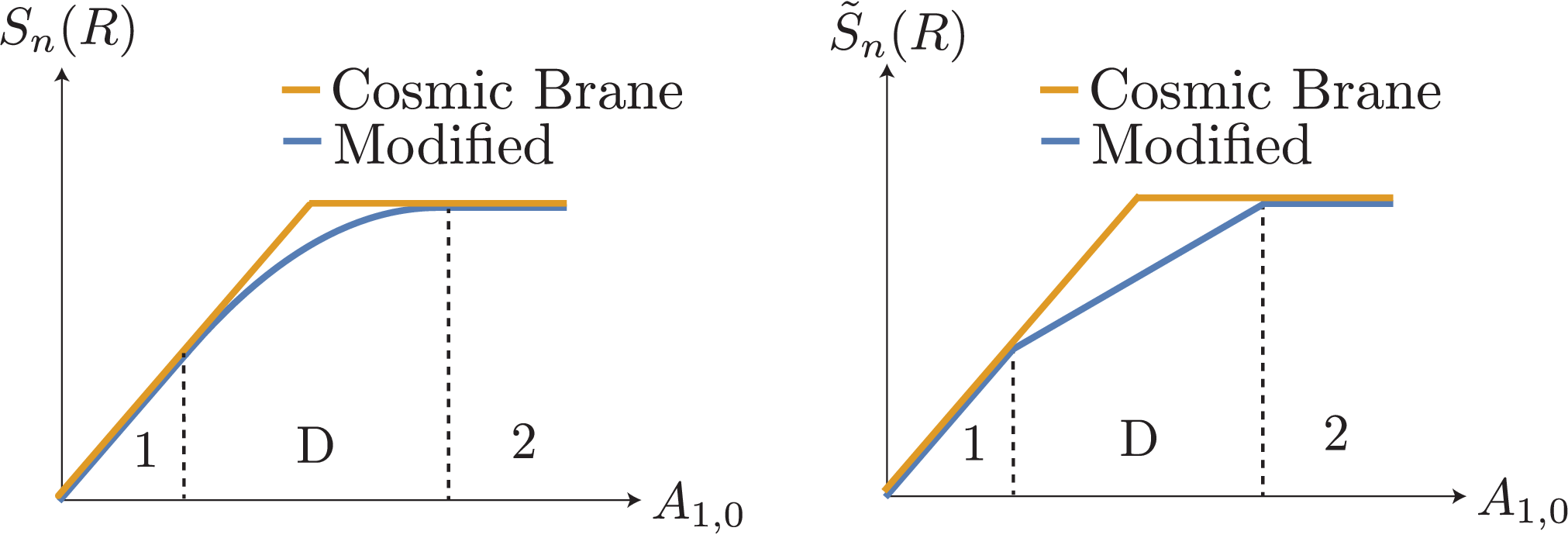}
\caption{The Renyi entropy (left) and refined Renyi entropy (right) computed by the modified cosmic brane proposal at a fixed generic value of $n<1$ is shown as a function of $A_{1,0}$ holding all other parameters fixed, along with the answer predicted by the original cosmic brane proposal. The three phases (1, 2 and D) arising as we move across the entanglement phase transition are labelled.}
\label{fig:gaussian2}
\end{figure}

Note that the existence of these corrections to the original cosmic brane proposal does not rely on being close to an entanglement phase transition, although the nearer we are to a phase transition, the corrections to the original cosmic brane proposal arise closer to $n=1$.
For illustration, in \figref{fig:gaussian2}, we depict the different phases that arise for a fixed generic value of $n<1$ as we increase $A_{1,0}$ across an entanglement phase transition, holding everything else fixed.
The Renyi entropy at $n<1$ is given by
\begin{align}
    S_n(R) = \begin{cases}
        S_n^{(1)}(R)&A_{1,0}<A_{2,0}-\frac{(1-n)\sg^2}{4nG},\\
        S_n^{(D)}(R)&A_{1,0}\in\[A_{2,0}-\frac{(1-n)\sg^2}{4nG},A_{2,0}+\frac{(1-n)\sg^2}{4nG}\],\\
        S_n^{(2)}(R)&A_{1,0}>A_{2,0}+\frac{(1-n)\sg^2}{4nG},
    \end{cases}
\end{align}
and the refined Renyi entropy follows the same pattern of phases.
The results are plotted in \figref{fig:gaussian2}.
Since $\sg=O(\sqrt{G})$ for gravitational states prepared by a smooth path integral, it is clear from this example that there is an $O(1)$ window of areas around the entanglement phase transition, where there are $O\(\frac{1}{G}\)$ corrections to the original cosmic brane proposal.

While we only discussed the simplest example where we can observe leading order corrections to the original cosmic brane proposal, the results for an arbitrary Gaussian are discussed in detail in Appendix~\ref{app:gaussian}.

Before moving on, we note that while this was a toy example, any smooth distribution can be approximated by a Gaussian near its peak.
This was used to analyze the entanglement entropy in Ref.~\cite{Marolf:2020vsi}, since it is universally a good approximation near $n=1$.
Similarly for us, while the details of the distribution will become important as the saddles move away from the peak, the results of this section are applicable in general (not necessarily Gaussian) states for $n\approx1$.

\section{Agreement with Known Results}
\label{sec:agree}

We will now provide evidence for our proposal by demonstrating consistency with known results for the Renyi entropy.
In \secref{sub:PSSY} and \secref{sub:eig}, we apply our formalism to two settings where previous results for the Renyi entropies exist even for $n<1$: that of the PSSY model and high energy eigenstates, respectively.
In each of these cases, large corrections to the Renyi entropy were found for $n<1$ and we shall reproduce this using our formalism. 
By showing consistency with these results, we provide evidence for our modified cosmic brane proposal which we derived by assuming a diagonal approximation.

\subsection{PSSY Model}
\label{sub:PSSY}

\begin{figure}[t]
\centering
\includegraphics[scale=0.7]{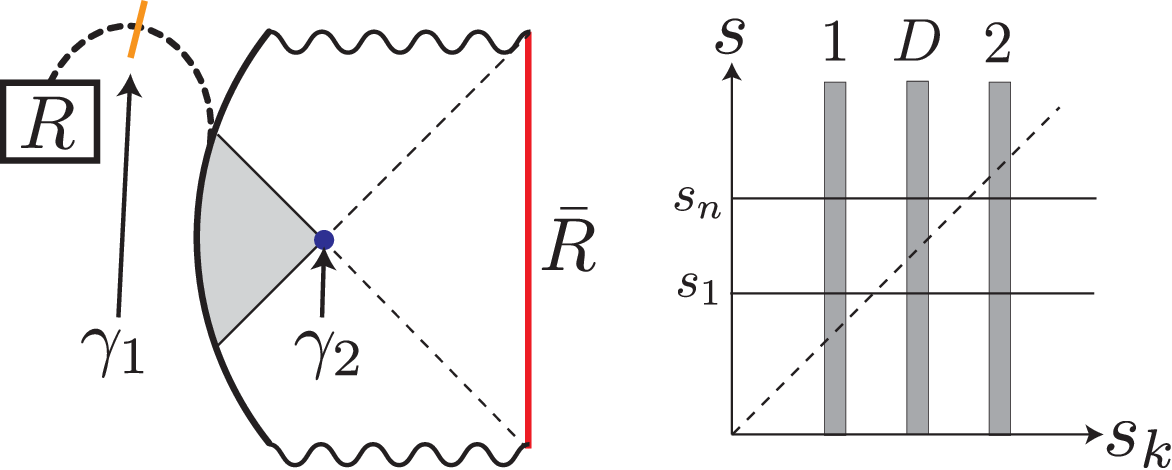}
\caption{(left): The PSSY model involves a JT black hole coupled to EOW branes. The $k$ flavours of the branes are entangled with $R$. The two (quantum) extremal surfaces are $\g_1$ (orange) and the horizon $\g_2$ (blue). (right): The non-zero region of the distribution $p\(s_k,s\)$ represented in terms of variables $\(s_k,s\)$ for the three different phases (labelled). The locations of the saddles $\tilde{s}^{(1)}(n)=s_1$ and $\tilde{s}^{(2)}(n)=s_n$ are marked for reference, assuming $n<1$.\label{fig:PSSY}}
\end{figure}

The PSSY model is a model of an evaporating black hole in JT gravity coupled to end-of-the-world (EOW) branes \cite{Penington:2019kki}. 
The EOW branes have $k$ flavours that are maximally entangled with an auxiliary radiation system $R$ (see \figref{fig:PSSY}).
The Renyi entropy was computed in Ref.~\cite{Akers:2022max}, with large corrections found for $n<1$ (see also Ref.~\cite{Dong:2021oad}).
Here, we show that our proposal reproduces the Renyi entropies precisely.

We will first need to review some basic facts about the model.
The partition function of JT gravity with an asymptotic boundary of renormalized length $\b$ and an EOW brane is given by \cite{Penington:2019kki}
\begin{equation}\label{eq:part}
    Z\(\b\) = \int_0^{\infty}ds\, \rho(s) y(s),
\end{equation}
where $\rho(s)$ is the density of states and $y(s)$ is the Boltzmann weight associated to the thermal spectrum at inverse temperature $\b$. 
For our analysis, we will work in the simplifying limit where the tension of the branes is chosen to be large. 
We then have $\rho(s) \approx e^{S_0+2\pi s}$ and $y(s)\approx e^{-\b s^2/2}$. 
The remaining free parameters in the theory are $k$, $S_0$ and $\b$.
The parameter $S_0$ will be taken to be large in order to suppress higher genus corrections.
Then, the semiclassical limit is controlled by $\b$ which can be rescaled to $\b G$ to restore the dependence on Newton's constant.

The candidate RT surfaces for subregion $R$ are $\g_1$ and $\g_2$ shown in \figref{fig:PSSY}.\footnote{Note that these are really quantum extremal surfaces, but we refer to them as RT surfaces for simplicity.}
To be precise, $\g_1$ is the trivial surface but includes a bulk contribution from the semiclassical entanglement between $R$ and the EOW branes.
Since the bulk entanglement is bipartite, we can follow the discussion in \secref{sec:renyi} to include bulk quantum corrections.
The bulk state already has a flat spectrum, and thus all we need to do is add a contribution of $\log k$ to the area operator at this surface, thus giving us $\frac{A_1}{4G}\equiv \log k$.
The surface $\g_2$ is the horizon of the black hole with area $A_2$ and, in this theory, $\frac{A_2}{4G}$ is interpreted as the value of the dilaton.
To compare with the notation of Ref.~\cite{Akers:2022max}, it will be convenient to parametrize the areas using $\(s_k,s\)$ where $\log k \equiv S_0 +2\pi s_k$ and $\frac{A_{2}}{4G}\equiv S_0+2\pi s$.

In order to apply the modified cosmic brane proposal to compute the Renyi entropy, the remaining ingredient is the probability distribution over areas, which we can obtain using our discussion in \secref{sub:wave}.
The distribution $p\(s_k,s\)$ has support only in a small window around a definite value of $s_k$ since the semiclassical entanglement spectrum is flat.
In order to compute the $s$-dependent part of $p\(s_k,s\)$, we can evaluate the action of a saddle with fixed $s$ and use \Eqref{eq:prob}.
This is straightforward since \Eqref{eq:part} already includes the contribution from all values of $s$ and we simply need to project it to a given value of $s$.
Thus the $s$-dependent part of $p\(s_k,s\)$ is
\begin{equation}
    p\(s_k,s\) \propto \exp\[-\frac{\b}{2}\(s-s_1\)^2\].
\end{equation}
This is a Gaussian peaked at $s=s_1$, which is the $n=1$ value of $s_n:=\frac{2\pi}{n\b}$ (not to be confused with $s_k$).
We depict this distribution for different values of $s_k$ in \figref{fig:PSSY}.

We can now apply the modified cosmic brane proposal for the Renyi entropy, which takes the form
\begin{equation}
    S_n\(R\)= \frac{1}{1-n}\max_{s}\( -\frac{n\b (s-s_1)^2}{2} +(1-n) \(S_0+2\pi \min\[s_k,s\]\)\).
\end{equation}
As discussed before, there are three potential maxima which in this case are given by
\begin{equation}
    \tilde{s}^{(1)}(n) = s_1, \qquad \tilde{s}^{(2)}(n) = s_n, \qquad s^{(D)}(n) = s_k,
\end{equation}
where we have used notation analogous to our previous examples to represent the possible phases and as before, the maxima are only allowed if they lie within the correct domain.

\begin{figure}[t]
\centering
\includegraphics[scale=0.4]{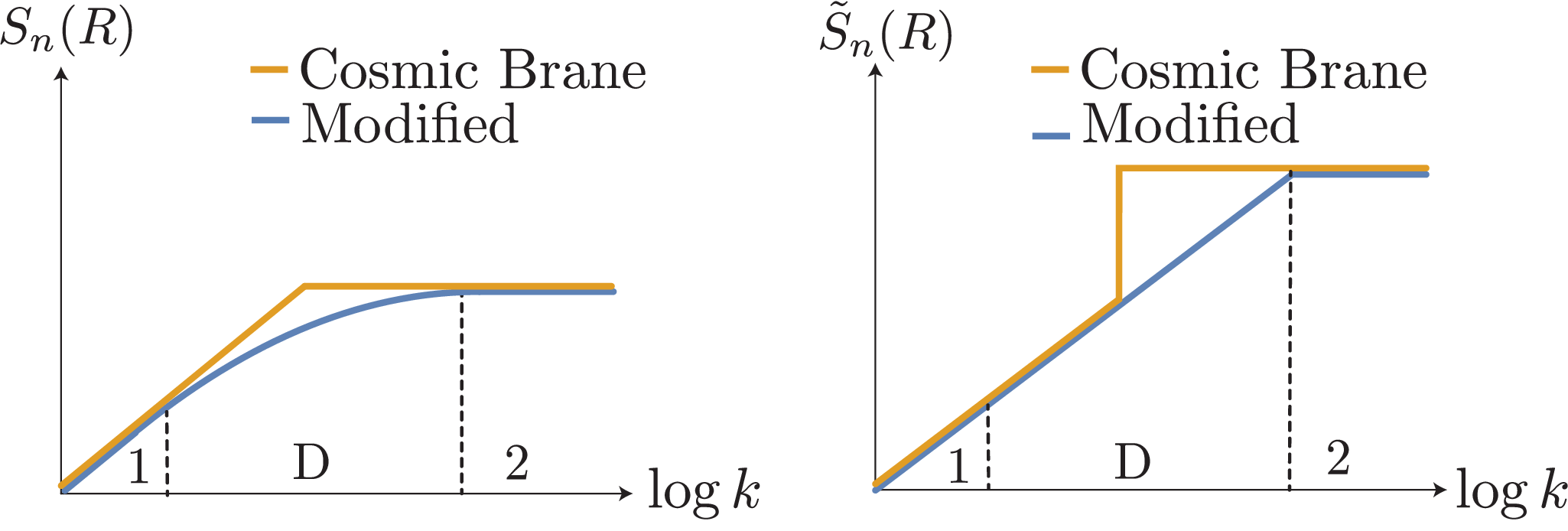}
\caption{The Renyi entropy (left) and refined Renyi entropy (right) computed by the modified cosmic brane proposal at a generic value of $n<1$ is shown as a function of $\log k$ holding all other parameters fixed, along with the answer predicted by the original cosmic brane proposal. The three phases (1, D and 2) arising as we move across the entanglement phase transition are labelled.}
\label{fig:n1}
\end{figure}

We can now divide our analysis into two cases: $n\geq 1$ and $n<1$.
\paragraph{\underline{$n\geq 1$}:}
For $n\geq1$, we obtain two phases, Phase 1 and Phase 2 respectively, i.e.,
\begin{align}
    S_n(R) = \begin{cases}
        \log k &s_k < \frac{\pi}{\b}\(1+\frac{1}{n}\)\\
        S_0+\frac{2\pi^2}{\b}\(1+\frac{1}{n}\)&s_k>\frac{\pi}{\b}\(1+\frac{1}{n}\)
    \end{cases},
\end{align}
which results from a discontinuous jump in the global maxima from $\tilde{s}^{(1)}$ to $\tilde{s}^{(2)}$ at $s_k=\frac{\pi}{\b}\(1+\frac{1}{n}\)$.
Correspondingly, the refined Renyi entropy is given by
\begin{align}
    \tilde{S}_n(R) = \begin{cases}
        \log k &s_k<\frac{\pi}{\b}\(1+\frac{1}{n}\)\\
        S_0+\frac{4\pi^2}{n\b}&s_k>\frac{\pi}{\b}\(1+\frac{1}{n}\)
    \end{cases}.
\end{align}
As expected, this is consistent with the original cosmic brane proposal since Phase D never appears.
\paragraph{\underline{$n<1$}:}
For $n<1$, we obtain all three phases, Phase 1, Phase D, and Phase 2 respectively, i.e.,
\begin{align}
    S_n(R) = \begin{cases}
        \log k &s_k<s_1\\
        S_0+\frac{2\pi s_k-\frac{n\b s_k^2}{2}-\frac{2\pi^2 n}{\b}}{1-n}&s_k\in\[s_1,s_n\]\\
        S_0+\frac{2\pi^2}{\b}\(1+\frac{1}{n}\)&s_k>s_n
    \end{cases},
\end{align}
which results from a continuous shift in the global maxima from $\tilde{s}^{(1)}$ to $s^{(D)}$ at $s_k=s_1$ and from $s^{(D)}$ to $\tilde{s}^{(2)}$ at $s_k=s_n$.
Correspondingly, the refined Renyi entropy is given by
\begin{align}
    \tilde{S}_n(R) = \begin{cases}
        \log k &s_k<s_n\\
        S_0+\frac{4\pi^2}{n \b}&s_k>s_n
    \end{cases}.
\end{align}
Both of these results are plotted in \figref{fig:n1} to contrast with the original cosmic brane proposal.

These are precisely the results obtained by Ref.~\cite{Akers:2022max} and we have reproduced them using the modified cosmic brane proposal.

\subsection{High Energy Eigenstates}
\label{sub:eig}

\begin{figure}[t]
\centering
\includegraphics[scale=0.7]{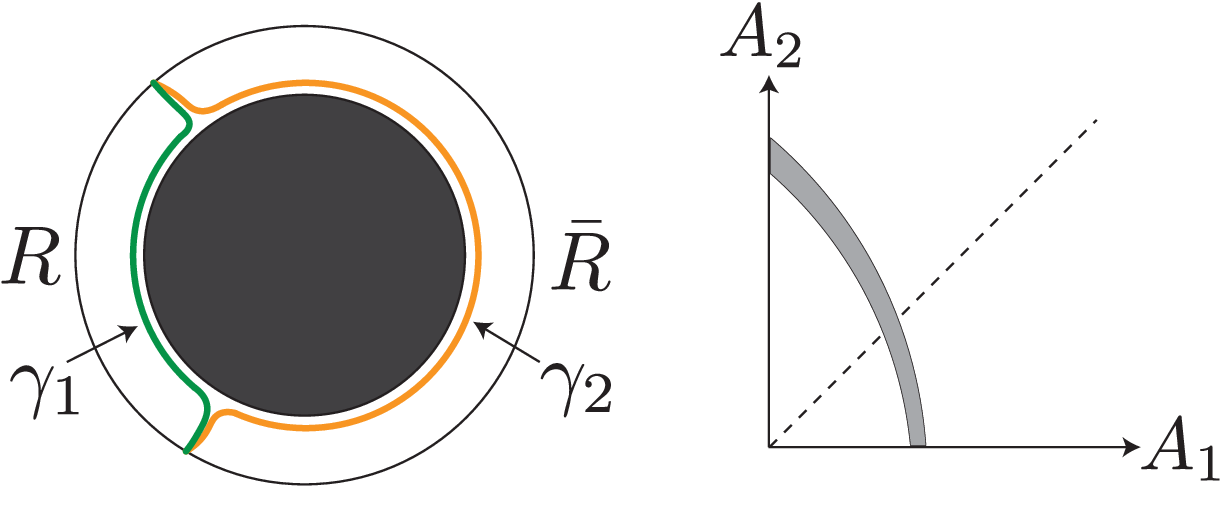}
\caption{(left): The holographic dual of a high energy eigenstate is a black hole geometry in the exterior. We are probing it in a limit where the black hole (shaded black) is large and reaches close to the boundary. The subregion $R$ which has two candidate RT surfaces $\g_1$ (green) and $\g_2$ (orange) which mostly hug the horizon. (right): The non-zero region of the distribution $p\(A_1,A_2\)$.\label{fig:energy}}
\end{figure}

We now consider high energy eigenstates of a single boundary holographic CFT.
The holographic dual of such states is a black hole geometry in the exterior.
In particular, we will be interested in the thermodynamic limit where the black hole is large and reaches close to the boundary as shown in \figref{fig:energy}.

Such a setup was studied for general chaotic theories in Ref.~\cite{Murthy:2019qvb} and was then studied in a holographic context in Ref.~\cite{Dong:2020iod}.
Here, we will apply the entanglement spectrum proposed in Ref.~\cite{Murthy:2019qvb} to holographic CFTs.
Their proposal was that for subregion $R$,
\begin{equation}\label{eq:sred}
    \tr\(\rho_R^n\) = \frac{\int dE_1 \, e^{S_{\text{min}}(E_1)}\[e^{-S(E)} e^{S_{\text{max}}(E_1)}\]^n}{e^{-S(E)}\int dE_1\, e^{S_{\text{min}}(E_1)}e^{S_{\text{max}}(E_1)}},
\end{equation}
where $S_{\text{min}(\text{max})}(E_1)=\min(\max)\[S_1\(E_1\),S_2\(E-E_1\)\]$ and $S_1,S_2$ are the thermodynamic entropies of subregions $R,\bar{R}$ at a given subsystem energy.

For a holographic theory, we can suggestively rewrite \Eqref{eq:sred} as
\begin{equation}\label{eq:sred2}
\begin{split}
    \tr\(\rho_R^n\) \approx \int dA_1 \,dA_2 \,\d\(E-E_1-E_2\) \tilde{p}\(E_1,E_2\)^n \exp\((1-n)\min\[\frac{A_1}{4G},\frac{A_2}{4G}\]\),
\end{split}
\end{equation}
where we have defined $\tilde{p}\(E_1,E_2\)=e^{-S(E)} e^{S_1(E_1)}e^{S_2(E_2)}$.
Moreover, we have used $S_1(E_1)\approx \frac{A_1}{4G}$ and $S_2(E_2)\approx \frac{A_2}{4G}$ since the geometry is approximately that of a large black hole in the exterior, identical to the thermal state.
In the thermodynamic limit, the areas of the RT surfaces are dominated by the volume-law term that comes from the portion hugging the horizon, purely determined by the subsystem energy.\footnote{For this discussion, we are subtracting out the divergent area contribution coming from the near boundary region, which is subleading in the thermodynamic limit where we take the volume large while keeping the cutoff finite.}
Since the area increases monotonically with energy, we have inverted this relation to implicitly express the subsystem energies $E_1$ and $E_2$ as functions of the areas $A_1$ and $A_2$.
The $\d$ function in \Eqref{eq:sred2} should be understood as a sharply peaked window function controlling the width of the microcanonical ensemble to obtain a semiclassical spacetime.
With this understanding, we can read off the probability distribution,
\begin{equation}
 p\(A_1,A_2\) \equiv \tilde{p}\(E_1 (A_1),E_2(A_2)\)\d\(E-E_1(A_1)-E_2(A_2)\),
\end{equation}
which is supported on a codimension-1 region in the $\(A_1,A_2\)$ parameter space as shown in \figref{fig:energy}.
In doing so, we have ignored any subleading contributions from the Jacobian as well as the window function.
Note that it should also be possible to derive the probability distribution directly from the gravitational path integral using the techniques of Ref.~\cite{Dong:2020iod}.

With this understanding, \Eqref{eq:sred} in the saddle point approximation is equivalent to our diagonal approximation \Eqref{eq:diag2}.
Consequently, our modified cosmic brane proposal which follows from \Eqref{eq:diag2} leads to identical results for the Renyi entropies obtained in Ref.~\cite{Murthy:2019qvb}.
In particular, the leading order corrections to the $n<1$ Renyi entropy seen in Ref.~\cite{Murthy:2019qvb} can be explained by our proposal.

\section{Discussion}
\label{sec:disc}

To summarize, we have offered a modified cosmic brane proposal to compute Renyi entropies in holographic systems that should be interpreted as an update to the Lewkowycz-Maldacena proposal \cite{Lewkowycz:2013nqa}.
Our modified proposal reproduces all previously known results for the holographic Renyi entropy: it always agrees with the original cosmic brane proposal for $n\geq 1$ and explains leading order corrections to the $n<1$ Renyi entropy found in various situations.
We now comment on various aspects of our work and possible future directions.

\begin{figure}[t]
\centering
\includegraphics[scale=0.6]{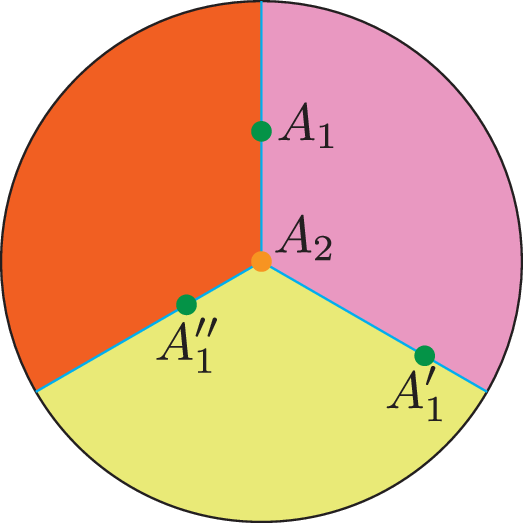}
\caption{A replica-symmetry breaking contribution to $\tr\(\rho_R^n\)$ (for $n=3$) that arises from the off-diagonal terms in \Eqref{eq:rho_diag}. The Euclidean saddle is built up from constituent fixed-area saddles and has conical defects (coloured circles) at the extremal surfaces. The surface $\g_2$ is identified and thus has a unique area, whereas the surfaces $\g_1$ have different areas in different copies. \label{fig:RSB}}
\end{figure}

A basic point we would like to emphasize here is that while we discussed RT surfaces in the paper, all of our results should naturally apply to HRT surfaces \cite{Hubeny:2007xt} in time-dependent settings as well.
Notably, the surfaces $\g_1$ and $\g_2$ are spacelike separated and thus, their areas can be simultaneously fixed even in time dependent situations.

An important future direction is to understand whether the off-diagonal terms in \Eqref{eq:rho_diag} can be shown to be negligible in general.
An example of the contribution of such a term is depicted in \figref{fig:RSB}.
It would be interesting to apply gravitational path integral arguments analogous to those in Ref.~\cite{Colafranceschi:2023txs} to prove that replica-symmetry breaking terms are indeed subleading.
We expect that this would provide a general justification for the Lewkowycz-Maldacena assumption.
We leave this analysis for future work.

Another important future direction is to understand bulk quantum corrections more generally, especially to $O(1)$.
As discussed in \secref{sec:renyi}, for a bipartitely entangled state, the bulk entropy simply modifies the definition of area and in this case, we know that the effect of replica-symmetry breaking terms is subleading.
For a general bulk state, if we ignore replica-symmetry breaking contributions then \Eqref{eq:rhoRfixed}, including bulk quantum corrections, becomes
\begin{align}
\begin{aligned}
    \tr\(\rho_R\(A_1,A_2\)^n\)&= \exp\(-I\[g_n\]+nI\[g_1\]\)\\
    &= \exp\((1-n)\min\[\frac{A_1}{4G}+S_{n,\text{bulk}}(\g_1),\frac{A_2}{4G}+S_{n,\text{bulk}}(\g_2)\]\)\label{eq:quantum},
\end{aligned}
\end{align}
where $S_{n,\text{bulk}}(\g_i)$ represents the bulk Renyi entropy for the entanglement wedge defined by $\g_i$ in the fixed-area state $|A_1,A_2\rangle$ defined by the gravitational path integral.
Using the diagonal approximation \Eqref{eq:diag1}, now including bulk quantum corrections, we get:
\begin{align}\label{eq:renyi_bulk}
    S_n^{MC}(R) =        \frac{1}{1-n}\log\[\sum_{A_1,A_2}\(p\(A_1,A_2\)^n \exp\[(1-n)\min_{i=1,2}\(\frac{A_i}{4G}+S_{n,\text{bulk}}(\g_i)\)\]\)\].
\end{align}
It is important in this case to use the diagonal approximation \Eqref{eq:diag1} instead of the saddle point approximation \Eqref{eq:diag2} since there are $O(\log G)$ corrections when we drop the sum.
It would be interesting in the future to analyze the bulk quantum corrections and see if the assumptions of diagonal approximation and replica symmetry that enter \Eqref{eq:renyi_bulk} are valid.

Another interesting aspect to be explored is the invariance of the Renyi entropies under bulk renormalization group (RG) flow \cite{Susskind:1994sm}.\footnote{See Refs.~\cite{Solodukhin:2011gn,Bousso:2015mna,Dong:2023bax} for some discussion of this issue.}
In particular, for $n<1$, we have seen that the diagonal saddle with $A_1=A_2$ can dominate even for such smooth states.
Such a saddle generically corresponds to conical defect opening angles different from $\frac{2\pi}{n}$ at the extremal surfaces.
While one abstractly expects invariance under bulk RG flow in this case as well, the details would require carefully defining the area operator with a UV cutoff and understanding how it evolves under RG flow.
It will be interesting to understand this in more detail in the future.

\begin{acknowledgments}

We are very grateful to Geoff Penington for extremely helpful discussions. 
XD is supported in part by the U.S. Department of Energy, Office of Science, Office of High Energy Physics, under Award Number DE-SC0011702 and by funds from the University of California. 
JKF is supported by the Marvin L.~Goldberger Member Fund at the Institute for Advanced Study and the National Science Foundation under Grant No. PHY-2207584. 
PR is supported in part by a grant from the Simons Foundation, by funds from UCSB, the Berkeley Center for Theoretical Physics; by the Department of Energy, Office of Science, Office of High Energy Physics under QuantISED Award DE-SC0019380, under contract DE-AC02-05CH11231 and by the National Science Foundation under Award Number 2112880. 
This material is based upon work supported by the Air Force Office of Scientific Research under award number FA9550-19-1-0360.

\end{acknowledgments}

\appendix
\section{Detailed Analysis of the Gaussian Example}
\label{app:gaussian}

We will now analyze the Gaussian distribution example of \secref{sec:gaussian} more exhaustively.
We remind the reader that the probability distribution over the areas is
\begin{equation}
    p\(A_1,A_2\) = \exp\[-\frac{1}{2}\(\vec{A}-\vec{A}_0\)\cdot C^{-1}\cdot \(\vec{A}-\vec{A}_0\)\],
\end{equation}
where $\vec{A}=\(A_1,A_2\)$, $\vec{A}_0=\(A_{1,0},A_{2,0}\)$ represents the area vector at the peak of the distribution, and $C$ is the covariance matrix given by
\begin{equation}
    C=\begin{pmatrix} 
  \sg_1^2     & r \sg_1 \sg_2\\ 
  r \sg_1 \sg_2 & \sg_2^2
\end{pmatrix},
\end{equation}
where $r\in\[-1,1\]$.

The maxima $\vec{\tilde{A}^{(1)}}(n)$ and $\vec{\tilde{A}^{(2)}}(n)$ are given by
\begin{align}
    \vec{\tilde{A}^{(1)}}(n) &= \vec{A}_0 - \frac{n-1}{4 G n}\begin{pmatrix} 
  \sg_1^2 \\ 
  r \sg_1 \sg_2 
\end{pmatrix},\\
\vec{\tilde{A}^{(2)}}(n) &= \vec{A}_0 - \frac{n-1}{4 G n}\begin{pmatrix} 
  r \sg_1 \sg_2 \\ 
  \sg_2^2 
\end{pmatrix}.
\end{align}
If these saddles dominate, then the corresponding Renyi entropies are given by
\begin{align}
    S_n^{(1)}(R) &= \frac{A_{1,0}}{4G} + \frac{(1-n)\sg_1^2}{32 n G^2}, \\
    S_n^{(2)}(R) &= \frac{A_{2,0}}{4G} + \frac{(1-n)\sg_2^2}{32 n G^2},
\end{align}
and the corresponding refined Renyi entropies are given by
\begin{align}
    \tilde{S}_n^{(1)}(R) &= \frac{A_{1,0}}{4G} + \frac{(1-n)\sg_1^2}{16 n G^2}, \\
    \tilde{S}_n^{(2)}(R) &= \frac{A_{2,0}}{4G} + \frac{(1-n)\sg_2^2}{16 n G^2}.
\end{align}
Finally, we have the candidate peak at $\vec{A^{(D)}}(n)$, given by
\begin{equation}
A_1^{(D)}=A_2^{(D)}=\frac{A_{1,0}\sg_2(\sg_2-r\sg_1)}{\sg_{-}^2} + \frac{A_{2,0}\sg_1(\sg_1-r\sg_2)}{\sg_{-}^2}+\frac{(1-n)(1-r^2)\sg_1^2 \sg_2^2}{4 n G \sg_{-}^2},
\end{equation}
where $\sg_{-}^2 = \sg_1^2+\sg_2^2-2 r \sg_1 \sg_2 \geq 0$.
The Renyi entropies in the diagonal phase are given by
\begin{equation}
    S_n^{(D)}(R) = -\frac{r\sg_1\sg_2\(A_{1,0}+A_{2,0}\)}{4G\sg_{-}^2}+\frac{\frac{2Gn\(\Delta A_0\)^2}{n-1}+A_{2,0}\sg_1^2+A_{1,0}\sg_2^2}{4G\sg_{-}^2}+\frac{\sg_1^2 \sg_2^2 (n-1) \left(r^2-1\right) }{32 G^2 n \sg_{-}^2},
\end{equation}
where $\Delta A_0=A_{2,0}-A_{1,0}$, and the corresponding refined Renyi entropies are
\begin{equation}
    \tilde{S}_n^{(D)}(R) = \frac{A_{1,0}\sg_2(\sg_2-r\sg_1)}{4G\sg_{-}^2} + \frac{A_{2,0}\sg_1(\sg_1-r\sg_2)}{4G\sg_{-}^2}+\frac{(1-n)(1-r^2)\sg_1^2 \sg_2^2}{16 n G^2 \sg_{-}^2}.
\end{equation}

\begin{figure}[t]
\centering
\includegraphics[scale=0.7]{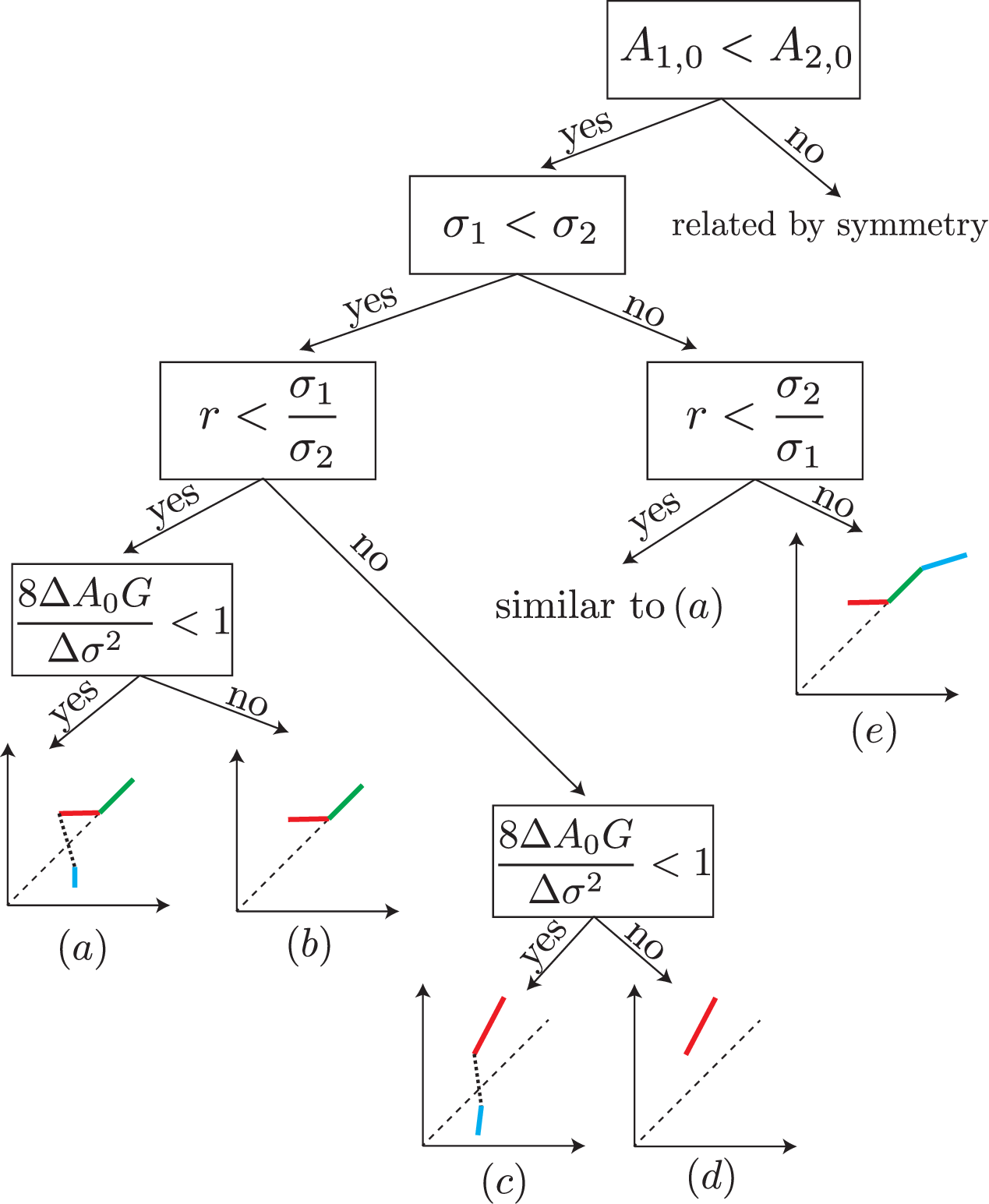}
\caption{The flowchart of all the qualitatively different cases is depicted. The schematic location on the $(A_1,A_2)$ plane of the dominant saddle as a function of $n$ is shown for each case with $\vec{\tilde{A}^{(1)}}(n)$ (red), $\vec{\tilde{A}^{(2)}}(n)$ (blue), and $\vec{A^{(D)}}(n)$ (green). Dotted black lines indicate a discontinuous jump in the peak.\label{fig:flow}}
\end{figure}

With these expressions in hand, we can now analyze the different possible cases.
Without loss of generality, consider the case where $\vec{A}_0$ lies in the domain $A_1<A_2$.
The qualitatively different cases that we can consider are shown in the flowchart in \figref{fig:flow}.

\paragraph{Case (a):} When $r<\frac{\sg_1}{\sg_2}$, there is a transition from Phase 1 to Phase D at a critical value $n^{(1)}_*\in\[0,1\]$ given by
\begin{equation}
    n^{(1)}_* = \frac{1}{1+\frac{4 G \Delta A_0}{\sg_1(\sg_1-r \sg_2)}}.
\end{equation}
On the other hand, if $8\Delta A_0 G< \Delta \sg^2$ where $\Delta \sg^2=\sg_2^2-\sg_1^2$, then there is a transition from Phase 1 to Phase 2 at a critical value $\tilde{n}_*>1$ given by
\begin{equation}
    \tilde{n}_* = \frac{1}{1-\frac{8 G \Delta A_0}{\Delta \sg^2}}.
\end{equation}
In summary, we have
\begin{align}
    S_n(R) = \begin{cases}
        S_n^{(2)}(R)&n>\tilde{n}_*\\
        S_n^{(1)}(R)&n\in[n^{(1)}_*,\tilde{n}_*]\\
        S_n^{(D)}(R)&n\in[0,n^{(1)}_*]
    \end{cases}.
\end{align}

\paragraph{Case (b):} When $r<\frac{\sg_1}{\sg_2}$ but $8\Delta A_0 G > \Delta \sg^2$, then the transition from Phase 1 to Phase 2 no longer happens since $\tilde{n}_*<0$.
Thus, we have
\begin{align}
    S_n(R) = \begin{cases}
        S_n^{(1)}(R)&n>n^{(1)}_*\\
        S_n^{(D)}(R)&n\in[0,n^{(1)}_*]
    \end{cases}.
\end{align}

\paragraph{Case (c):} This is an interesting case where we can potentially lose the saddle $\vec{\tilde{A}^{(1)}}(n)$ at a critical value $n^{(1)}_*>1$.
However, as long as $8\Delta A_0 G< \Delta \sg^2$, there is a transition to the saddle $\vec{\tilde{A}^{(2)}}(n)$ which starts dominating at $n=\tilde{n}_*\in\[1,n^{(1)}_*\]$ before the Phase 1 saddle is lost.
No transitions happen at $n<1$.
Thus, we have
\begin{align}
    S_n(R) = \begin{cases}
        S_n^{(1)}(R)&n\in\[0,\tilde{n}_*\]\\
        S_n^{(2)}(R)&n>\tilde{n}_*
    \end{cases}.
\end{align}

\paragraph{Case (d):} In this case, there are no transitions and we have Phase 1 for all values of $n\in\[0,\infty\]$.

\paragraph{Case (e):} If $r \sg_1 > \sg_2$ (note that this requires $\sg_1>\sg_2$ since $r<1$), then we have two transitions.
At $n=n^{(1)}_*<1$, there is a transition from Phase 1 to Phase D.
Then, at another critical value $n=n^{(2)}_*<n^{(1)}_*$ given by
\begin{equation}
     n^{(2)}_* = \frac{1}{1-\frac{4 G \Delta A_0}{\sg_2(\sg_2-r \sg_1)}},
\end{equation}
there is a continuous transition from Phase D to Phase 2.
Thus, we have
\begin{align}
    S_n(R) = \begin{cases}
        S_n^{(1)}(R)&n>n^{(1)}_*\\
        S_n^{(D)}(R)&n\in[n^{(2)}_*,n^{(1)}_*]\\
        S_n^{(2)}(R)&n\in[0,n^{(2)}_*]
    \end{cases}.
\end{align}

\addcontentsline{toc}{section}{References}
\bibliographystyle{JHEP}
\bibliography{apssamp}

\end{document}